\documentclass[a4paper,11pt]{article}
\usepackage[final]{hyperref}
\usepackage{amsmath,amssymb,amsthm,mathtools,cleveref,tikz-cd,tikz-feynhand,mleftright,physics,orcidlink,accsupp,enumerate}
\usepackage[noadjust]{cite}
\usepackage[smalltableaux]{ytableau}
\usepackage[osf]{newpxtext}
\usepackage[vvarbb,slantedGreek]{newpxmath}
\usepackage[final,spacing]{microtype}
\usepackage[british]{babel}
\newcommand*{\acctextsc}[1]{%
  \BeginAccSupp{%
    ActualText=\detokenize{#1},%
    method=escape,
  }%
  \textsc{\lowercase{#1}}%
  \EndAccSupp{}%
}
\newcommand{\fg}{\mathfrak{g}}
\newcommand{\ttr}{\mathtt{r}}
\newcommand{\tte}{\mathtt{e}}
\newcommand{\ttd}{\mathtt{d}}
\newcommand{\mf}{\mathfrak}
\newtheorem{lemma}{Lemma}
\newtheorem{theorem}{Theorem}
\theoremstyle{definition}
\newtheorem{definition}{Definition}
\newcommand\email[1]{\href{mailto:#1}{\nolinkurl{#1}}}
\newcommand\ZZ{\mathbb Z}
\newcommand\CC{\mathbb C}
\renewcommand\delta\deltaup
\renewcommand\varepsilon\varepsilonup

\hypersetup {
    pdftitle = {Homotopy representations of extended holomorphic symmetry in holomorphic twists},
    pdfauthor = {David Simon Henrik Jonsson, Hyungrok Kim, and Charles Alastair Stephen Young},
    pdfsubject = {hep-th, math-ph},
    pdflang = {en-GB-oed},
    pdfkeywords = {Holomorphic field theory, holomorphic twist, L∞-algebra, supersymmetric Yang–Mills theory}
}
\title{Homotopy representations of extended holomorphic symmetry in holomorphic twists}
\author{D. Simon H. Jonsson\orcidlink{0009-0001-7155-8496}\footnotemark[1]\\\email{d.jonsson@herts.ac.uk}\\Hyungrok Kim\orcidlink{0000-0001-7909-4510}\footnotemark[1]\\\email{h.kim2@herts.ac.uk}\\Charles A. S. Young\orcidlink{0000-0002-7490-1122}\footnote{Department of Mathematics and Theoretical Physics, University of Hertfordshire, Hatfield, Hertfordshire\ \textsc{al10 9ab}, United Kingdom}\\\email{c.young8@herts.ac.uk}}
\begin{document}
\maketitle
\begin{abstract}
We argue that holomorphic twists of supersymmetric field theories naturally come with a symmetry \(L_\infty\)-algebra that nontrivially extends holomorphic symmetry. This symmetry acts on spacetime fields only up to homotopy, and the extension is only visible at the level of higher components of the action. We explicitly compute this for the holomorphic twist of ten-dimensional supersymmetric Yang--Mills theory, which produces a nontrivial action of a higher \(L_\infty\)-algebra on (a graded version) of five-dimensional affine space.
\end{abstract}
\tableofcontents

\section{Introduction and summary}
Twisting of physical theories \cite{Witten:1988ze,2011arXiv1111.4234C,Elliott:2018cbx,elliott2020taxonomy} has attracted great interest in the physics literature.
In particular, the pure spinor formalism \cite{Howe:1991bx,Berkovits:2002uc,Eager_2022,Elliott_2023,jonssonthesis,cederwall2024canonical} (see reviews in \cite{Cederwall:2013vba,Cederwall:2022fwu}), which naturally describes such theories as supergravity \cite{Cederwall:2009ez,Cederwall:2010tn}, supersymmetric Yang--Mills theory \cite{Berkovits:2001rb,Cederwall:2001bt} and \acctextsc{M}2-brane models \cite{Cederwall:2008vd,Cederwall:2008xu,Cederwall:2009ay},
accommodates twisting naturally \cite{saberi2021twisting,Hahner:2023kts}.

Physical theories come with representations of spacetime symmetry algebras, such as (super-)Poincaré algebras and (super-)conformal algebras. It has been long known that for theories with more than four supercharges it is often difficult to manifest this symmetry `off shell', that is, without using equations of motion. The pure spinor formalism provides a means of producing off-shell supermultiplets by introducing appropriate infinite towers of auxiliary fields and furthermore shows that the on-shell supermultiplets in fact carry a \emph{homotopy representation} of the spacetime symmetries; the higher components of the action then correspond to the equations of motion needed to make the symmetry algebra close.

In this paper, we argue that holomorphic twists of supersymmetric field theories naturally come with more than just the holomorphic symmetry but rather a certain \(L_\infty\)-extension of holomorphic symmetry. The extension is not visible at the level of strict representations, but spacetime fields naturally form a homotopy representation of this extended symmetry. We shall treat in detail the example of the holomorphic twist of ten-dimensional supersymmetric Yang--Mills theory. 
This twisted theory is holomorphic Chern--Simons theory on $\CC^5$ \cite{Baulieu:2010ch,elliott2020taxonomy,saberi2021twisting}, which enjoys a manifest $\mf{isl}(5) = \mf{sl}(5) \ltimes \mathbf 5$ symmetry. 
As we shall see, it naturally comes with the extended holomorphic symmetry \(L_\infty\)-algebra
\begin{equation}\label{eq:tildeisl}\operatorname{\widetilde{\mf{isl}}}(5)\coloneqq    \left(\mathfrak{sl}(5)\ltimes\mathbf{10}\xrightarrow00\xrightarrow0\mathbf{5}\right)
\end{equation}
equipped with a certain higher bracket \(\mu_4\); and this \(L_\infty\)-algebra acts on $\CC[z^1,z^2,z^3,z^4,z^5]$ (with appropriate grading) in the $L_\infty$-algebraic sense.
This may be seen as a non-strict \(L_\infty\)-algebra action of \(\operatorname{\widetilde{\mf{isl}}}(5)\) on the (graded version of) five-dimensional complex affine space \(\mathbb A^5\).

We work with minimal models (of both the symmetry algebra and the field content), which canonically separates the physical information and makes clear the presence of higher-order structures (\(L_\infty\)-algebras and their representations), rather than a larger strict model, which is not canonical and mixes in the physical degrees of freedom together with the unphysical auxiliary fields; this ensures that all information that we recover is physical and independent of the choice of auxiliary fields.

One way to think about this is to recall that twisting is akin to dimensional reduction \cite{pirsa_PIRSA:23070028} in which, rather than eliminating dependence on bosonic coordinates, we eliminate dependence on fermionic coordinates (restrict to \(Q\)-closed fields for a supersymmetry \(Q\)), which results in the `pair annihilation' of bosonic and fermionic coordinates. From this perspective, we have an `as above, so below' heuristic: the actions of twisted theories resemble those of their twistings, just like dimensional reduction preserves the forms of actions.
Using the pure spinor formalism, ten-dimensional supersymmetric Yang--Mills theory may be formulated as a holomorphic Chern--Simons theory on a complex (21|16)-dimensional pure spinor superspace (with 10 complexified ordinary spacetime coordinates, 16 ordinary superspace fermonic coordinates, and 11 bosonic pure spinor coordinates). The twisted theory has the same form of a Chern--Simons theory, but this time on 5|0 dimensions, where we have killed 16|16 coordinates. Under this `dimensional reduction', the ten-dimensional \(\mathcal N=1\) super-Poincaré symmetry, which is (55|16)-dimensional, reduces to a (39|0)-dimensional extended holomorphic symmetry. This dimensional reduction corresponds to twisting the supersymmetry algebra and taking the minimal models of the symmetry algebra and its homotopy representation on the field content. The additional factor \(\mathbf{10}\) in \eqref{eq:tildeisl} and the concomitant \(\mu_4\) are the `dimensionally reduced' remnants of ten-dimensional super-Poincaré symmetry.

The discussion of the present paper is limited to the kinematics, that is, ignoring interactions and considering the linearized theory. This is not an essential restriction. A discussion of the interaction terms should make use of the \(L_\infty\)-algebra formalism \cite{Jurco:2018sby,Macrelli:2019afx,Jurco:2019yfd} for scattering amplitudes;
after colour-stripping, we should get a \(C_\infty\)-algebra \cite{Borsten:2021hua}, on which the extended holomorphic symmetry should act, forming an example of an open--closed homotopy algebra \cite{Kajiura:2004xu,Kajiura:2005sn,Kajiura:2006mt}. This, however, we leave to a future work.

While we focus on ten-dimensional supersymmetric Yang--Mills theory as a special case, the discussion is generic and applies, in principle, to the twists of any supersymmetric field theory. However, the twists in other dimensions often produce either a strict representation (with the \(\mathfrak{\widetilde{isl}}(d)\)-representation factoring through an \(\mathfrak{isl}(d)\)-representation)
or a higher representation of an \(L_\infty\)-superalgebra on affine superspace (with odd coordinates); \(\mathbb A^5\) is one of the few nontrivial purely bosonic examples that carry a higher symmetry. (For more discussion, see \cref{sec:higher_d}.)

All of our discussion is classical; there may be obstructions to quantization in the form of anomalies. For our main example of the holomorphic twist of ten-dimensional supersymmetric Yang--Mills theory, the twist (five-dimensional holomorphic Chern--Simons theory) is known to have anomalies unless it is coupled in a consistent fashion to Kodaira--Spencer gravity \cite{costello2015quantization,costello2020anomaly}.

Local operators in a holomorphic theory are expected to form higher analogues of vertex algebras \cite{Beem:2013sza,Saberi:2019fkq,Bomans:2023mkd,Gaiotto:2024gii}. Although the additional \(L^{\wedge2}\) symmetry that we find does not seem to be part of a higher Virasoro algebra (since it is not part of holomorphic symmetries), it may arise as modes of some local operator, in which case it will be part of a higher vertex algebra, and the \(\mu_4\) that we find may be part of the higher brackets of the higher vertex algebra.

\subsection{Organization of this paper}
This paper is organized as follows. \Cref{sec:background} reviews the generalities of twisting \(L_\infty\)-algebras and modules over them and the appearance of higher components of the spacetime symmetry algebras and higher components of nonstrict representations of \(L_\infty\)-algebras, both in the untwisted and twisted cases. \Cref{sec:10d} then computes the higher components of the representation of supersymmetry for ten-dimensional supersymmetric Yang--Mills theory, the higher products of the corresponding twisted extended holomorphic algebra, and the higher components of its representation on the twisted supermultiplet. \Cref{sec:higher_d} briefly surveys phenomena that appear in dimensions other than ten.

In the body of the paper, we will usually refer to irreducible representations of \(\mathfrak{sl}(5)\) using their Dynkin labels, supplemented by Young tableaux where they are helpful.

\section{Mathematical background}\label{sec:background}
Here we briefly review the relevant concepts of twisting of \(L_\infty\)-algebras and their modules. For more detailed reviews, see \cite{Loday:2012aa,dotsenko2019twisting,kraft2022introduction,Dotsenko_2023}.

\subsection{\texorpdfstring{\(L_\infty\)}{L∞}-algebras}
An \(L_\infty\)-algebra is a homotopy generalization of the concept of a Lie algebra.
\begin{definition}
An \(L_\infty\)-algebra \((\mathfrak g,\{\mu_k\}_{k\geq 1})\) consists of a graded vector space \(\mathfrak g=\bigoplus_{i\in\mathbb Z}\mathfrak g^i\) together with skew-symmetric, multilinear maps \(\mu_k\colon\mathfrak g^{\wedge k}\to\mathfrak g\) of degree \(2-k\) for \(k\in\{1,2,3,\dotsc\}\) that satisfy the identity
\begin{equation}\label{eq:homotopy_Jacobi_identities}
    0=\sum_{\substack{i+j=n\\\sigma\in\operatorname{Sh}(i,j)}}(-1)^j\chi(\sigma,x)\mu_{j+1}(\mu_i(x_{\sigma(1)},\dotsc,x_{\sigma(i)}),\dotsc,x_{\sigma(i+j)})=0.
\end{equation}
In the above, \(\operatorname{Sh}(j_1,\dotsc,j_k)\) denotes the collection of \emph{shuffles}, which are permutations \(\sigma\) of \(\{1,\dotsc,j_1+\dotsb+j_k\}\) such that \(\sigma(1)<\dotsb<\sigma(j_1)\) and \(\sigma(j_1+1)<\dotsb<\sigma(j_1+j_2)\) and so on up to \(\sigma(j_1+\dotsb+j_{k-1}+1)<\dotsb<\sigma(j_1+\dotsb+j_k)\).
The symbol \(\chi(\sigma,x)\) denotes the skew-symmetric Koszul sign
\begin{equation}
    x_1\wedge\dotsb\wedge x_{k}
    =\chi(\sigma,x) x_{\sigma(1)}\wedge\dotsb\wedge x_{\sigma(k)},
  \end{equation}
  defined for homogeneous elements $x_1,\ldots,x_k\in\mf g$ inside the exterior algebra $\bigwedge^\bullet\mf g$.
\end{definition}
In what follows, we will often leave the products $\{\mu_k^{\mathfrak g}\}_{k\geq 1}$ implicit, and simply refer to an $L_\infty$-algebra through its underlying graded vector space. The identities \eqref{eq:homotopy_Jacobi_identities} imply that $\mu_1\circ\mu_1=0$ so that $\mf{g}$ is in particular a cochain complex.

\begin{definition}
  A morphism of $L_\infty$-algebras $\phi\colon(\fg,\{\mu_k^{\mathfrak g}\}_{k\geq 1})\rightsquigarrow(\mf{h},\{\mu_k^{\mathfrak h}\}_{k\geq 1})$ consists of skew-symmetric, multilinear component maps
  \begin{equation}
    \label{eq:inftymorphism}
    \phi^{(n)}\colon\fg^{\wedge n}\to \mf h
  \end{equation}
of degree $1-n$ for $n\in \{1,2,\dots\}$, satisfying the following coherence relations:
  \begin{multline}
  \label{eq:coherencemorphisms}
      \sum_{\mathclap{\substack{j\in\{1,\dotsc,i\}\\k_1+\dotsb+k_j=i\\\sigma\in\operatorname{Sh}(k_1,\dotsc,k_j)}}}
      \frac{\zeta(\sigma,k,x)}{j!}
      \mu_j^{\mathfrak h}\big(\phi^{(k_1)}( x_{\sigma(1)}, \dotsc, x_{\sigma(k_1)}),
      \dotsc, \phi^{(k_j)}( x_{\sigma(k_1+\dotsb+k_{j-1}+1)}, \dotsc,  x_{\sigma(i)})\big)\\
      =
      \sum_{\mathclap{\substack{j+k=i\\\sigma\in\operatorname{Sh}(j,k)}}}(-1)^k \chi(\sigma,x) \phi^{(k+1)}\big(\mu_j^{\mathfrak g}( x_{\sigma(1)}, \dotsc,  x_{\sigma(j)}),  x_{\sigma(j+1)}, \dotsc,  x_{\sigma(i)}\big),
\end{multline}
where
\begin{equation}
    \zeta(\sigma,k,x)\coloneqq\chi(\sigma,x)(-1)^{\sum_{1 \le m<n \le j} k_m k_n+\sum_{m=1}^{j-1} k_m(j-m)+\sum_{m=2}^j\left(1-k_m\right)\sum_{k=1}^{k_1+\dotsb+k_{m-1}}|x_{\sigma(k)}|} .
  \end{equation}
We shall sometimes omit the $\mu_k$ and just write $\fg\rightsquigarrow\mf h$. 
$L_\infty$-morphisms compose associatively, so that one has the category whose objects are $L_\infty$-algebras and whose morphisms are $L_\infty$-morphisms between them.

An $L_\infty$-morphism is an \emph{ $L_\infty$-(quasi-)isomorphism} if the first component map is a (quasi-)isomorphism of the underlying cochain complexes.
\end{definition}
\paragraph{Homotopy transfer of $L_\infty$-algebras.} 
\(L_\infty\)-algebras admit a good homotopy theory in the sense that minimal models exist and can be computed by homotopy transfer using a strong deformation retract. Let us sketch how this works. Concretely, given an $L_\infty$-algebra $(\mf g,\{\mu_k^{\mathfrak g}\}_{k\geq 1})$ one can always choose a \emph{strong deformation retract}, denoted by a triple $(i,p,h)$, from the underlying cochain complex $(\mf g,\mu_1)$ to its cohomology $\operatorname H(\mf g)$:
\begin{equation}
\begin{tikzcd}
    (\mathfrak g,\mu_1) \rar[shift left, "p"] \ar[loop left, "h"] & (\operatorname H(\mathfrak g),0) \lar[shift left, "i"]
\end{tikzcd}
\label{sdr}
\end{equation}
(i.e.~\(pi=\operatorname{id}_{\operatorname H(\mf g)}\) and \(ip=\operatorname{id}_{\mf g}-[d,h]\)).
Then there exists an \(L_\infty\)-algebra structure on the cohomology \(\operatorname H(\mathfrak g)\) together with an \(L_\infty\)-quasi-isomorphism
\begin{equation}\label{map to min}
    e\colon\operatorname H(\mathfrak g)\rightsquigarrow\mathfrak g,
\end{equation}
whose first component is \(e^{(1)}=i\); furthermore, there exist explicit formulae for the \(L_\infty\)-algebra structure of \(\operatorname H(\mathfrak g)\) and the quasi-isomorphism \(e\) in terms of \((i,p,h)\) \cite{Loday:2012aa}, e.g.~using the tensor trick \cite{Berglund_2014}, which can be interpreted as a sum over Feynman diagrams \cite{Macrelli:2019afx,Saemann:2020oyz}. For example the ternary bracket $\mu_3^{\operatorname H(\mf g)}$, is (modulo relative signs) the sum

\begin{equation}
  \begin{tikzpicture}[scale=0.5,baseline={([yshift=-1ex]current bounding box.center)}]
    \draw [thick] (-2,0) -- (-0.6,-1.4);
    \draw [thick] (-0.35,-1.65) -- (0,-2);
    \draw [thick] (2,0) -- (0,-2);
    \draw [thick] (0,-2) -- (0,-3);
    \draw [thick] (-1,-1)--(0,0);
    \node at (-.5,-1.5) {\tiny{$h$}};
    \node at (0,-3.5) {\tiny\(p\)};
    \node at (0,0.5) {\tiny\(i\)};
    \node at (-2,0.5) {\tiny\(i\)};
    \node at (2,0.5) {\tiny\(i\)};
    
     \fill[gray!50] (-1,-1) circle (0.35cm);
    \draw [thick] (-1,-1) circle (0.35cm);
    \node at (-0.98,-1) {\tiny$\mu_2$};
    \fill[gray!50] (0,-2) circle (0.35cm);
    \draw [thick] (0,-2) circle (0.35cm);
    \node at (0.02,-2) {\tiny$\mu_2$};
  \end{tikzpicture}
  +
  \begin{tikzpicture}[scale=0.5,baseline={([yshift=-1ex]current bounding box.center)}]
    \draw [thick] (-2,0) -- (0,-2);
    \draw [thick] (2,0) -- (0.65,-1.35);
    \draw [thick] (0,-2) -- (0.4,-1.6);
    \draw [thick] (0,-2) -- (0,-3);
    \draw [thick] (1,-1)--(0,0);
    \node at (0.55,-1.45) {\tiny{$h$}};
     \fill[gray!50] (1,-1) circle (0.35cm);
    \draw [thick] (1,-1) circle (0.35cm);
    \node at (1.02,-1) {\tiny$\mu_2$};
    \fill[gray!50] (0,-2) circle (0.35cm);
    \draw [thick] (0,-2) circle (0.35cm);
    \node at (0.02,-2) {\tiny$\mu_2$};\node at (0,-3.5) {\tiny\(p\)};
    \node at (0,0.5) {\tiny\(i\)};
    \node at (-2,0.5) {\tiny\(i\)};
    \node at (2,0.5) {\tiny\(i\)};
  \end{tikzpicture}
  +
  \begin{tikzpicture}[scale=0.5,baseline={([yshift=-1ex]current bounding box.center)}]
    \draw [thick] (-2,0) -- (-0.6,-1.4);
    \draw [thick] (-0.35,-1.65) -- (0,-2);
    
    \draw [thick] (1,-1) -- (0,-2);
    \draw [thick] (0,-3) -- (0,-2);
    \draw [thick] (1,-1)--(0,0);
    \draw [thick] (-1,-1) -- (0.3,-0.57);
    \draw [thick] (2,0) -- (0.7,-.45);
    \node at (-.5,-1.5) {\tiny{$h$}};
  
    \node at (0,-3.5) {\tiny\(p\)};
    \node at (0,0.5) {\tiny\(i\)};
    \node at (-2,0.5) {\tiny\(i\)};
    \node at (2,0.5) {\tiny\(i\)};

     \fill[gray!50] (-1,-1) circle (0.35cm);
    \draw [thick] (-1,-1) circle (0.35cm);
    \node at (-0.98,-1) {\tiny$\mu_2$};
    \fill[gray!50] (0,-2) circle (0.35cm);
    \draw [thick] (0,-2) circle (0.35cm);
    \node at (0.02,-2) {\tiny$\mu_2$};
  \end{tikzpicture}
  +
  \begin{tikzpicture}[scale=0.5,baseline={([yshift=-1ex]current bounding box.center)}]
    \draw [thick] (-2,0) -- (0,-2);
    \draw [thick] (2,0) -- (0,-2);
    \draw [thick] (0,0) -- (0,-2);
    \draw [thick] (0,-2) -- (0,-3);
    \node at (0,-3.5) {\tiny\(p\)};
    \node at (0,0.5) {\tiny\(i\)};
    \node at (-2,0.5) {\tiny\(i\)};
    \node at (2,0.5) {\tiny\(i\)};
    \fill[gray!50] (0,-2) circle (0.35cm);
    \draw [thick] (0,-2) circle (0.35cm);
    \node at (0.02,-2) {\tiny$\mu_3$};
  \end{tikzpicture}.
\end{equation}
More generally, $\mu_k^{\operatorname H(\mf g)}$ is computed by a sum\footnote{The explicit relative signs between the trees can be worked out by using the aforementioned tensor trick, for example.} over all rooted trees with $k$ leaves, where one decorates the leaves with $i$, the $n+1$-ary vertices with $\mu_n$, the internal edges with $h$, and the root with $p$.
The $L_\infty$-algebra structure on \(\operatorname H(\mf g)\) is called the \emph{minimal model} of \(\mathfrak g\); minimal models are unique up to $L_\infty$-isomorphisms.

\paragraph{Twisting \(L_\infty\)-algebras}
\(L_\infty\)-algebras also admit a notion of \emph{twist} with respect to a Maurer--Cartan element; for reviews, see \cite{dotsenko2019twisting,kraft2022introduction,Dotsenko_2023}. In the definitions below, for a $L_\infty$-algebra $(\mf g,\{\mu_k\}_{k\geq 1})$, we assume for simplicity that \(\mu_i=0\) for sufficiently large \(i\); this can be relaxed \cite{kraft2022introduction}.
\begin{definition}[\cite{Jurco:2018sby,Dotsenko_2023}]
Let \((\mathfrak g,\{\mu_k^{\mathfrak g}\}_{k\geq 1})\) be an \(L_\infty\)-algebra such that \(\mu_i=0\) for sufficiently large \(i\).
A \emph{Maurer--Cartan element} \(Q\in\mathfrak g^1\) of \(\mathfrak g\) is an element of degree \(1\) such that
\begin{equation}
    \sum_{i=1}^\infty\frac1{i!}\mu_i(Q,\dotsc,Q)=0.
\end{equation}
\end{definition}
\begin{definition}[\cite{Dotsenko_2023}]
Let \((\mathfrak g,\{\mu_k^{\mathfrak g}\}_{k\geq 1})\) be an \(L_\infty\)-algebra such that \(\mu_i=0\) for sufficiently large \(i\).
Let \(Q\in\mathfrak g^1\) be a Maurer--Cartan element.
The \emph{twist} of \(\mathfrak g\) with respect to \(Q\) is the \(L_\infty\)-algebra \(\mathfrak g_Q\) whose underlying graded vector space is that of \(\mathfrak g\) but whose brackets \(\mu_k^Q\) are
\begin{equation}\label{eq:twistedproducts}
\begin{aligned}
    \mu_k^Q&\colon&\mathfrak g_Q^{\wedge k}&\to\mathfrak g_Q\\
    &&(x_1,\ldots,x_k)&\mapsto\sum_{i\geq 0}\frac1{i!}\mu_{i+k}(Q,\ldots,Q,x_1,\dotsc,x_k).
\end{aligned}
\end{equation}
\end{definition}

\subsection{\texorpdfstring{\(L_\infty\)}{L∞}-representations}
The notion of a representation of (or module over) a Lie algebra generalizes to the setting of homotopy algebras as follows.
\begin{definition}[\cite{Lada:1994mn,lada,Jurco:2018sby}]
  \label{def:inftyrep}
    An \emph{$L_\infty$-representation}
of an \(L_\infty\)-algebra \((\mathfrak g,\{\mu_k^{\mathfrak g}\}_{k\geq 1})\) on a graded vector space \(M\) is an \(L_\infty\)-algebra structure \(\{\mu_k^{\mathfrak g\ltimes M}\}_{k\geq 1}\) on the direct sum \(\mathfrak g\oplus M\) 
such that
    \begin{equation}
        \mu_k^{\mathfrak g\ltimes M}(x_1\oplus0,\dotsc,x_k\oplus0)=\mu_k(x_1,\dotsc,x_k)
    \end{equation}
     and  \(\mu_{k+1}^{\mathfrak g\ltimes M}(x_1,\ldots x_k,m)\in 0\oplus M\) for \(x_1,\dotsc,x_k\in\mathfrak g\) and $m\in M$, and such that $\mu_{k}^{\mathfrak g\ltimes M}$ vanishes whenever at least two of its arguments belong to \(0\oplus M\subset\mathfrak g\oplus M\). We will write $\mf g \ltimes M$ to refer to $L_\infty$-algebras of this form.
We write
    \begin{equation}
        \rho^{(k)}(x_1,\dotsc,x_{k})\coloneqq\mu_{k+1}^{\mathfrak g\ltimes M}(x_1,\dotsc,x_{k},-)\colon M\to M.
      \end{equation}
      Observe that, in particular, \(\rho^{(k)}\) carries degree \(1-k\).
    Note that \(\rho^{(0)}\) defines a differential on \(M\), making it a cochain complex. We call an $L_\infty$-representation \emph{strict} whenever $\rho^{(k)}=0$ for $k>1$.
\end{definition}
The \(L_\infty\)-algebra homotopy Jacobi identities \eqref{eq:homotopy_Jacobi_identities} then can be written as a series of coherence relations amongst the \(\rho^{(k)}\)'s and $\mu_j^{\mf g}$'s.

As with \(L_\infty\)-algebras themselves, \(L_\infty\)-representations admit a good homotopy theory in that minimal models exist and homotopy transfer is possible. That is, given an \(L_\infty\)-algebra $\mathfrak g$ and a \(\mf g\)-representation \(M\),
we can always choose a strong deformation retract
\begin{equation}
  \label{eq:retractrep0}
    \begin{tikzcd}
        \ar[loop left,"{(h,h')}"] (\mathfrak g\oplus M,\mu_1+\rho^{(0)} )\rar[shift left, "{(p,p')}"] &  \lar[shift left, "{(i,i')}"]  (\operatorname H(\mathfrak g)\oplus\operatorname H(M),0)
\end{tikzcd}
\end{equation}
and perform homotopy transfer of $L_\infty$-algebra structures along this retract to obtain an $L_\infty$-algebra on $\operatorname H(\mf g)\oplus \operatorname H(M)$, which then defines the \(L_\infty\)-representation of \(\operatorname H(\mathfrak g)\) on \(\operatorname H(M)\).\footnote{The induced brackets on $\operatorname H(\mf g)\oplus \operatorname H(M)$ automatically satisfy the conditions given in \cref{def:inftyrep}. Indeed, as there are no brackets in $\mf g \oplus M$ that reduce the number of factors of $M$, no such brackets can arise through composition.}\footnote{This is the minimal model for the two-coloured operad of pairs of \(L_\infty\)-algebras and their \(L_\infty\)-representations, rather than the minimal model for the uncoloured operad of \(L_\infty\)-representations over a fixed \(L_\infty\)-algebra \(\mathfrak g\).}

Given an \(L_\infty\)-algebra \(\mathfrak g\) and a $\mf g$-representation \(M\) with structure maps \(\rho^{(k)}\),
then it is clear by inspection that a Maurer--Cartan element \(Q\in\mathfrak g\) is also Maurer--Cartan element of \(\mathfrak g\ltimes M\) and that the twist \((\mathfrak g\ltimes M)_Q\) factorizes as \((\mathfrak g\ltimes M)_Q=\mathfrak g_Q \ltimes M_Q\) \cite{kraft2022introduction}, where
\(M_Q\) comes with the structure maps
\begin{equation}\label{eq:twistedmodule}
    \rho_Q^{(k)}(x_1,\dotsc,x_k) \coloneqq \sum_{i=0}^\infty\frac1{i!}\rho^{(i+k)}(Q,\dotsc,Q,x_1,\dotsc,x_k).
  \end{equation}

\section[Higher symmetry of twisted 10D SYM]{Higher symmetry of twisted ten-dimensional supersymmetric Yang--Mills theory}\label{sec:10d}
In the Batalin--Vilkovisky formalism \cite{Batalin:1977pb,1981PhLB..102...27B,1983PhRvD..28.2567B,BATALIN1984106,Batalin:1985qj},
the field content of a perturbative gauge theory is a graded vector space \((\mathcal F,d_{\mathcal F})\) that comes equipped with a differential.
Field theories often respect symmetry algebras such as the super-Poincaré algebra, the (super-)conformal algebra, the (super-)(anti-)de~Sitter algebra, the (super-)Galilean algebra, etc. 
The action of such a symmetry algebra may be \emph{off shell} (i.e.~on \(\mathcal F\) itself) or merely \emph{on shell} (i.e.\ only on the space of solutions to the equations of motion\footnote{Since we ignore interactions, for us this is the linearized equations of motion, but in general one should consider the interacting case.}). The symmetry algebras are usually ungraded or \(\mathbb Z/2\mathbb Z\)-graded (i.e.~superalgebras), and correspondingly the field space is \(\mathbb Z/2\mathbb Z\)-graded into bosons and fermions (in addition to the \(\mathbb Z\)-grading corresponding to ghost number).
This \(\mathbb Z/2\mathbb Z\) grading may be often lifted to a \(\mathbb Z\) grading; correspondingly, the \(\mathbb Z/2\mathbb Z\) grading of the field space may also be lifted to a \(\mathbb Z\)-grading.\footnote{In order to lift the $\ZZ/2\ZZ$ grading of field space to \(\mathbb Z\), we work with polynomials over spacetime rather than smooth functions. This simplification can be avoided; see \cite{Eager_2022}.} 
The \(\mathbb Z\)-grading enables a good homotopy theory of \(L_\infty\)-algebras and \(L_\infty\)-representations, and in particular off-shell realizations of symmetries can be in most cases lifted to a non-strict \(L_\infty\)-representation of the corresponding symmetry \cite{Eager_2022}.

We turn now to our main example of interest, which is the twist of ten-dimensional super-Yang--Mills theory. (For simplicity and convenience with twisting, we assume all symmetries and fields to be complexified.)

We first discuss the twisted super-Poincaré algebra itself.
Let 
\begin{equation} 
V \cong \CC^{10} 
\end{equation} 
be a ten-dimensional complex vector space equipped with a nondegenerate symmetric bilinear form.
The ten-dimensional \(\mathcal N=(1,0)\) super-Poincaré superalgebra is the Lie superalgebra
\begin{equation}
    \mathfrak o(V)\ltimes(\Piup S_+\oplus V)
\end{equation}
where \(S_\pm\) are the two \(16\)-dimensional Weyl spinor representations of \(\mathfrak o(V)\) and \(\Piup\) denotes parity reversal.
The \(\mathbb Z/2\mathbb Z\) grading of the super-Poincaré superalgebra can be lifted to a \(\mathbb Z\)-grading as the graded Lie algebra
\begin{equation}
    \mathfrak p \coloneqq \mathfrak o(V)\ltimes(S_+[-1]\oplus V[-2])
\end{equation}
in which the elements are graded as twice the conformal dimension (i.e.~rotations in degree \(0\), supertranslations in degree \(1\), translations in degree \(2\)).
For convenience, we can pick a basis \(\ttr_{\mu\nu}\), \(\ttd_\alpha\), \(\tte_\mu\) of \(\mathfrak p^0\), \(\mathfrak p^1\), \(\mathfrak p^2\) respectively. Then the structure constants for \([\mathfrak p^1,\mathfrak p^1]\subset\mathfrak p^2\) are
\begin{equation}
    [\ttd_\alpha,\ttd_\beta]=2\gamma^\mu_{\alpha\beta}\tte_\mu,
\end{equation}
where \(\gamma^\mu_{\alpha\beta}\) are the chiral gamma (or Pauli) matrices in ten dimensions, i.e.~the branching for the \(\mathfrak o(V)\)-representation \(S_+\otimes S_+\to V\).

\subsection{The minimal model of the holomorphic twist algebra}
Suppose we pick a nonzero Maurer--Cartan element of \(\mathfrak p\), i.e.\ a nonzero \(Q=Q^\alpha\ttd_\alpha\in S_+\) such that \([Q,Q]=0\), that is,
\begin{equation}
    \gamma^\mu_{\alpha\beta}Q^\alpha Q^\beta=0.
\end{equation}
This picks out a subspace
\begin{equation}
    L = [Q,S_+]\subset V.
  \end{equation}
This subspace $L$ is a maximal isotropic subspace with respect to the bilinear form on \(V\).
Indeed, using the Fierz identity 
\begin{equation}    2\gamma^\mu_{\alpha(\beta}\gamma^{\phantom0}_{\mu|\gamma)\delta}=-\gamma^\mu_{\beta\gamma}\gamma^{\phantom0}_{\mu\alpha\delta},
\end{equation}
we have
\begin{equation}
    Q^\beta Q^\gamma\gamma^\mu_{\alpha\beta}\gamma^{\phantom0}_{\mu\gamma\delta}\propto(Q^\beta Q^\gamma\gamma^\mu_{\beta\gamma})\gamma^{\phantom0}_{\mu\alpha\delta}=0;
\end{equation}
given now any elements \(\psi,\chi\in S_+\), consider the elements \([Q,\psi],[Q,\chi]\in L\). We have
\begin{equation}
    [Q,\psi]^\mu[Q,\chi]_\mu =(Q^\beta Q^\gamma\gamma_{\alpha\beta}^\mu \gamma_{\mu \gamma\delta})\psi^\alpha\chi^\delta=0.
\end{equation}
Thus $L$ is indeed contained in its own orthogonal complement, i.e.\ it is an isotropic subspace.
Isotropy implies \(\dim L\le5\). Furthermore, the set of Maurer--Cartan elements of \(\mathfrak p\) consists of two \(\mathfrak o(V)\)-orbits, namely nonzero ones and \(\{0\}\); and it can be shown that, when \(Q\ne0\), then \(\dim L=5\) \cite{cdb3b74a-b986-3a9f-9841-aad15f0544ec}. That is, \(L\) is indeed a maximal isotropic subspace.

Thus, we have the short exact sequence of vector spaces
\begin{equation}\label{eq:decomposition_of_V}
    0\to L\to V\xrightarrow q L^*\to 0
\end{equation}
where the quotient \(q\) is via the composition \(V\xrightarrow\sim V^*\twoheadrightarrow L^*\) in which \(V\xrightarrow\sim V^*\) is given by the bilinear form on \(V\). Let us choose a splitting of \eqref{eq:decomposition_of_V} to write
\begin{equation}
    V=L\oplus L^*.
\end{equation}
This decomposition then fixes Lie subalgebras \(\mathfrak{sl}(L)\subset\mathfrak{gl}(L)\subset\mathfrak o(V)\), under which the ten-dimensional representation \(V\) canonically decomposes into \(\mathfrak{sl}(L)\) irreducible representations as
\begin{equation}
    V\cong_{\mf{sl}(L)} L\oplus L^*\cong_{\mf{sl}(L)} (1000)_{\mf{sl}(L)} \oplus (0001)_{\mf{sl}(L)}
\end{equation}
(Here and elsewhere we write $(ijkl)_{\mf{sl}(L)}$ for the irreducible representation with these $\mf{sl}(L)$ Dynkin labels, i.e.\ for the irreducible representation with highest weight $i\omega_1 + j \omega_2 + k \omega_3 + l\omega_4$ where $\omega_1,\dots,\omega_4$ are the fundamental weights.)

Similarly, the adjoint representation of \(\mf o(V)\) decomposes into irreducible \(\mf{sl}(L)\)-representations as
\begin{equation}
\begin{aligned}
    \mathfrak o(V) &\cong_{\mf{sl}(L)} \CC \oplus \mathfrak{sl}(L)\oplus L^{\wedge2}\oplus(L^*)^{\wedge2}\\
    &\cong_{\mf{sl}(L)} (0000)_{\mf{sl}(L)} \oplus (1001)_{\mf{sl}(L)} \oplus (0100)_{\mf{sl}(L)} \oplus (0010)_{\mf{sl}(L)},
\end{aligned}
\end{equation}
and the spinor representations \(S_\pm\) decompose as
\begin{align}
    S_+ &\cong_{\mf{sl}(L)} \bigoplus_{i=0}^2(L^*)^{\wedge(2i)}\cong_{\mf{sl}(L)}(0000)_{\mf{sl}(L)}\oplus(0010)_{\mf{sl}(L)}\oplus(1000)_{\mf{sl}(L)},\\
    S_- &\cong_{\mf{sl}(L)} \bigoplus_{i=0}^2(L^*)^{\wedge(2i+1)} \cong_{\mf{sl}(L)} S_+^*\cong(0001)_{\mf{sl}(L)}\oplus(0100)_{\mf{sl}(L)}\oplus(0000)_{\mf{sl}(L)}.
\end{align}
Our chosen element $Q\in S_+$ spans the one-dimensional $\mf{sl}(L)$-submodule $\CC\cong (L^*)^{\wedge0}$.
The twist of $\mf p$ by $Q$ is \cite[Prop.~3.3]{saberi2021twisting}
\begin{equation}
\mathfrak p_Q = 
\left(\begin{tikzcd}[row sep=0]
\underset{\mathclap{(0000)_{\mf{sl}(L)}}}{\CC} \rar["\operatorname{id}"] & \underset{(0000)_{\mf{sl}(L)}}{(L^*)^{\wedge0}[-1]} \\
\underset{(0010)_{\mf{sl}(L)}}{(L^*)^{\wedge2}} \rar["\operatorname{id}"] & \underset{(0010)_{\mf{sl}(L)}}{(L^*)^{\wedge2}[-1]} \\
& \underset{(1000)_{\mf{sl}(L)}}{(L^*)^{\wedge4}[-1]} \rar["\operatorname{id}"] & \underset{(1000)_{\mf{sl}(L)}}{L[-2]} \\
\underset{\mathclap{(1001)_{\mf{sl}(L)}}}{\mathfrak{sl}(L)}\,\,\ltimes\,\,\underset{\mathclap{(0100)_{\mf{sl}(L)}}}{\smash{L^{\wedge2}}\vphantom{\mathfrak{sl}(L)}}&&\underset{(0001)_{\mf{sl}(L)}}{L^*[-2]}
\end{tikzcd}\right).
\label{pQ}\end{equation}

We work in an explicit basis \((\ttr^i{}_j,\ttr^{ij},\ttr_{ij},\ttd,\ttd_{ij},\ttd^i,\tte^i,\tte_i)\) of \(\mathfrak p\) given in \cref{app:conventions}. In particular, the basis elements $\ttr^i{}_j$ span $\mf{gl}(L)$, and we shall write 
\begin{equation}\label{eq:define_sl(L)_basis}
    \tilde\ttr^i{}_j\coloneqq\ttr^i{}_j - \frac15\delta^i_j\ttr^k{}_k
\end{equation}
for the basis elements of $\mf{sl}(L)$. 
 
For a reason to become apparent in the theorem below, let us note that the representation $L^* \otimes (L^{\wedge 2})^{\wedge 3}$ of $\mf{sl}(L)$ decomposes into irreducibles as
\begin{equation}   
L^* \otimes (L^{\wedge 2})^{\wedge 3} \cong (0021)_{\mf{sl}(L)} \oplus (0110)_{\mf{sl}(L)} \oplus (1001)_{\mf{sl}(L)} \oplus (2010)_{\mf{sl}(L)} \oplus (2002)_{\mf{sl}(L)}. 
\end{equation}
In particular, the adjoint representation $\mf{sl}(L) \cong (1001)_{\mf{sl}(L)}$ occurs with multiplicity one. We shall write 
\begin{equation}  P_{L^*\otimes(L^{\wedge2})^{\wedge3}\to\mathfrak{sl}(L)} 
\end{equation}
for the projector onto this irreducible component. 
\begin{theorem}\label{thm:H(p_Q)}
The minimal model of the ten-dimensional twisted \(\mathbb Z\)-graded \(\mathcal N=(1,0)\) super-Poincaré algebra \(\mathfrak p_Q\) is the \(L_\infty\)-algebra whose underlying graded Lie algebra is
\begin{equation}
    \operatorname H(\mathfrak p_Q) = \mathfrak{sl}(L)\ltimes\left( L^{\wedge2}\oplus L^*[-2]\right),
\end{equation}
and whose higher brackets \(\mu_i\) \((i\ge3)\)
are all zero except for \(\mu_4\), whose only nonvanishing component is given by
\begin{equation}
\begin{aligned}
    \mu_4(\tte_i,\ttr^{jk},\ttr^{lm},\ttr^{np})
    &=-\delta_i^{[j}\varepsilon^{k]lmqr}\delta^{[n}_q\tilde\ttr^{p]}{}_r
    +\delta_i^{[j}\varepsilon^{k]npqr}\delta^{[l}_q\tilde\ttr^{m]}{}_r\\
    &\qquad-\delta_i^{[l}\varepsilon^{m]npqr}\delta^{[j}_q\tilde\ttr^{k]}{}_r
    +\delta_i^{[l}\varepsilon^{m]jkqr}\delta^{[n}_q\tilde\ttr^{p]}{}_r\\
    &\qquad-\delta_i^{[n}\varepsilon^{p]jkqr}\delta^{[l}_q\tilde\ttr^{m]}{}_r
    +\delta_i^{[n}\varepsilon^{p]lmqr}\delta^{[j}_q\tilde\ttr^{k]}{}_r\\
    &\eqqcolon P^{jklmnp;}_i{}^r_q\tilde\ttr^q{}_r,
\end{aligned}
\end{equation}
where \(P^{jklmnp;}_i{}^r_q\) is the projector
\begin{equation}\label{eq:mu4_projector}
    L^*\otimes(L^{\wedge2})^{\wedge3}\to\mathfrak{sl}(L),
\end{equation}
and where the skew-symmetrizations are unnormalized.
\end{theorem}
\begin{proof}
From \cref{pQ}, we see that there is an evident \(\mathfrak{sl}(L)\)-equivariant strong deformation retract $(i,p,h)$ of cochain complexes 
\begin{equation}
  \label{eq:ogretract}
    \begin{tikzcd}
        \ar[loop left,"h"] (\mf p_Q,\mathrm{ad}_Q) \rar[shift left, "p"] &  \lar[shift left, "i"]  (\operatorname H(\mf p_Q),0)
\end{tikzcd}
\end{equation}
from $\mf p_Q$ to its cohomology
\begin{equation}
\operatorname H(\mathfrak p_Q) = (\mathfrak{sl}(L)\ltimes\textstyle L^{\wedge2}\xrightarrow 0 0\xrightarrow 0 L^*).
\end{equation}
(The remaining $\mf{sl}(L)$ irreducible representations present in \cref{pQ} participate in trivial pairs; one defines the homotopy $h$ to act as the inverse to the differential on these.)
The Lie algebra structure $\mu_2$ on this cohomology $\operatorname{H}(\mf p_Q)$ is given by restriction.

It remains to check what higher brackets \(\mu_i\) are induced by homotopy transfer. 
We are to sum over rooted binary trees in which each vertex corresponds to the binary bracket $\mu_2^{\mf p}(-,-) = [-,-]$ of $\mf p$ (and thus of $\mf p_Q$), each internal edge to the homotopy $h$, each leaf to $i$ and the root to $p$ \cite{Loday:2012aa}.

We will use Feynman-diagrammatic terminology, referring to elements as `states' (see \cite{Macrelli:2019afx,Saemann:2020oyz}). 
Recall our notation \((\ttr^i{}_j,\ttr^{ij},\ttr_{ij},\ttd,\ttd_{ij},\ttd^i,\tte^i,\tte_i)\) for the basis elements of \(\mathfrak p\) (whose underlying graded vector space we identify with that of \(\mathfrak p_Q\)) as given in \cref{app:conventions}. We will refer to a state as \emph{intermediate} if it lies in the image of $h$.

Using the strong deformation retract \((i,p,h)\), we try to construct the possible intermediate states, keeping track of the representations under \(\mathfrak{sl}(L)\). Using the embedding $i\colon\operatorname H(\mf{p}_Q)\hookrightarrow \mf{p}_Q$, we will identify the basis elements \(\tilde\ttr^i{}_j,\ttr^{ij},\tte_i\) in $\mf{p}_Q$ with those of $\operatorname H(\mf{p}_Q)$. The products $\mu_i$ for $i>2$ can then be computed in a top-down recursive fashion by starting with two elements $x, y\in \operatorname{H}(\mf{p}_Q)$, and then compute which intermediate states are allowed by considering
\begin{equation}
  \label{eq:intermediate}
  h[i(x),i(y)].
\end{equation}
The next intermediate states are then computed by plugging \eqref{eq:intermediate}, and one element $a\in i(\operatorname{H}(\mf{p}_Q))\oplus \Im(h[i,i])$, into $[-,-]$. If the result lies in $\operatorname H(\mf{p}_Q)$, we apply $p$, and we are done. If not, we apply $h$ to obtain new intermediate states and then continue the procedure of pairing (using $[-,-]$) the newly obtained intermediate states with each other or with previously obtained intermediate states or states in the cohomology.

Starting with two elements of $\operatorname H(\mf{p}_Q)$, applying \(h[-,-]\) can only yield the intermediate states
\begin{equation}\label{eq:interm1}
    h[\tte_k,\ttr^{ij}]=\delta^i_k\ttd^j-\delta^j_k\ttd^i.
\end{equation}
Using \(\ttd^i\) together with \(\tilde\ttr^i{}_j,\ttr^{ij},\tte_i\), we can only further create 
\begin{subequations}
\begin{align}
[\ttd^i,\ttd^j] &= 0, \\
[\ttd^i,\tte_j] &= 0, \\
h[\ttr^{ij},\ttd^k] &= -\frac12\varepsilonup^{ijklm}\ttr_{lm},\\
[\tilde\ttr^i{}_j,\ttd^k]&=-\frac15\delta^i_j\ttd^k+\delta^k_j\ttd^i.
\end{align}
\end{subequations}
Among these, \([\tilde{\ttr}^i{}_j,\ttd^k]\) does not belong to the cohomology, i.e.\ it is not $Q$-closed, nor can it produce a new intermediate state since \(h([\tilde{\ttr}^i{}_j,\ttd^k])=0\).
Thus the only intermediate state we can create is given by \(h[\ttr^{ij},\ttd^k]\propto\ttr_{lm}\). Applied to \eqref{eq:interm1}, we obtain
\begin{equation}\label{eq:interm2}
    h[h[\tte_k,\ttr^{ij}],\ttr^{lm}]=h[\delta^i_k\ttd^j-\delta^j_k\ttd^i,\ttr^{lm}]
    =\frac12\left(
    \delta^i_k\varepsilon^{jlmnp}
    -\delta^j_k
    \varepsilon^{ilmnp}\right)\ttr_{np}.
\end{equation}
Using the new intermediate state \(\ttr_{ij}\) together with the previously created intermediate state \(\ttd^i\) and the cohomology \((\tilde\ttr^i{}_j,\ttr^{ij},\tte_i)\), we can create the following new states:
\begin{subequations}
\begin{align}
[\ttr_{ij},\tte_k] &= 0, \\
[\ttr_{ij},\ttd^k] &= 0, \\
[\tilde\ttr^i{}_j,\ttr_{kl}]&=
-\delta^i_k\ttr_{jl}-\delta^i_l\ttr_{kj}
+\frac25\delta^i_j\ttr_{kl}, \\
[\ttr^{ij},\ttr_{kl}] &= \delta^i_k\ttr^j{}_l-\delta^j_k\ttr^i{}_l-\delta^i_l\ttr^j{}_k+\delta^j_l\ttr^i{}_k, \\
[\ttr_{ij},\ttr_{kl}] &= 0.
\end{align}
\end{subequations}
Now, all these states sit in degree $0$ and are thus killed by $h$, so that none of them can create further intermediate states. The only nontrivial thing we can now do is to project to the cohomology: \([\tilde\ttr^i{}_j,\ttr_{kl}]\) never lies in the cohomology, whereas the traceless part of \([\ttr^{ij},\ttr_{kl}]\) does. Applied to \eqref{eq:interm2}, we obtain
\begin{equation}
\begin{aligned}
p[h[h[\tte_k,\ttr^{ij}],\ttr^{lm}],\ttr^{qr}]
    &=-\frac12\left(\delta^j_k
    \varepsilon^{ilmnp}
    -
    \delta^i_k\varepsilon^{jlmnp}\right)[\ttr_{np},\ttr^{qr}]\\
    &=\left(
    \delta^j_k
    \varepsilon^{ilmnp}
    -\delta^i_k\varepsilon^{jlmnp}
    \right)
    \left(
        \delta^q_n\ttr^r{}_p
        -\delta^r_n\ttr^q{}_p
    \right)\\
    &=\delta^{[j}_k
    \varepsilon^{i]lmnp}\delta^{[q}_n\tilde\ttr^{r]}{}_p.
\end{aligned}
\end{equation}
After graded-skew-symmetrization among the three arguments of the form \(\ttr^{ij}\), this yields the only nonvanishing component of \(\mu_4\).

There are no other \(\mu_i\) since we have systematically constructed all possible nonzero tree Feynman diagrams (by constructing all possible intermediate states that occur in them).
\end{proof}

One may doubt whether the nonstrictness and existence of a 4-bracket in \(\operatorname H(\mathfrak p_Q)\) is model-independent (i.e.~holds for all minimal models) or an accidental feature of the specific minimal model in question. By the general theory of minimal models, $\operatorname H(\mf p_Q)$ is unique up to $L_\infty$-isomorphisms. Concretely, we may ask whether there exists a strict minimal model of \(\operatorname H(\mathfrak p_Q)\) (hence with no higher brackets). The answer is no.
\begin{lemma}
  Let $\mf h$ be a minimal strict graded $L_\infty$-algebra defined on the graded vector space $\operatorname H(\mf p_Q)\cong \mf{sl}(L)\oplus L^{\wedge 2}\oplus L^*$. There exists no $L_\infty$-isomorphism $\phi\colon\operatorname H(\mf p_Q)\rightsquigarrow \mf h$.
\end{lemma}
\begin{proof}
    Suppose to the contrary that such an \(L_\infty\)-isomorphism $\phi\colon\operatorname H(\mf p_Q)\rightsquigarrow \mf h$ exists with component maps $\phi^{(k)}$. Without loss of generality, we may identify the underlying graded vector spaces of \(\operatorname H(\mf p_Q)\) and \(\mf h\) via $\phi^{(1)}$. Since $\operatorname H(\mf p_Q)$ (and hence $\mf h$) are concentrated in even degrees, even-order components of \(\phi\) (which have odd degree) vanish: $\phi^{(2k)}=0$. The coherence relations \cref{eq:coherencemorphisms} then implies that $\phi^{(1)}$ is a Lie-algebra isomorphism of the underlying graded Lie algebra structures on $\operatorname H(\mf p_Q)$ and $\mf h$. Moreover, the coherence relation on four elements reads
  \begin{multline}
    \phi^{(1)}(\mu_4^{\operatorname H(\mf p_Q)}(\tte_p,\ttr^{ij},\ttr^{kl},\ttr^{mn}))=\mu_2^{\mf h}(\phi^{(1)}(\tte_p),\phi^{(3)}(\ttr^{ij},\ttr^{kl},\ttr^{mn}))\\+\text{permutations}.
  \end{multline}
  Now, the left-hand side is nonzero and lies in the copy of $\mf{sl}(L)$ inside $\mf h$. But since $\mu_2^{\mf h}$ agrees with $\mu_2^{\operatorname H(\mf p_Q)}$, by virtue of $\phi^{(1)}$ being a Lie algebra morphism, we have that $\mu_2^{\mf h}$ is $\mf{sl}(L)$-equivariant, and hence the right-hand side cannot lie in $\mf{sl}(L)$, a contradiction.
\end{proof}

\subsection{Action on \texorpdfstring{\(\mathbb A^5\)}{𝔸⁵}}
The pure spinor formalism \cite{Eager_2022,Elliott_2023} associates certain sheaves on (a derived replacement of) the variety of Maurer--Cartan elements to off-shell representations of the super-Poincaré algebra.
In particular, for the ten-dimensional \(\mathcal N=(1,0)\) super-Poincaré algebra, it associates to the structure sheaf of the Maurer--Cartan variety \(\operatorname{Spec}\mathbb C[\lambda^\alpha]/(\gamma^\mu_{\alpha\beta}\lambda^\alpha\lambda^\beta)\) the pure spinor supermultiplet \cite{Eager_2022}
\begin{equation}
\begin{aligned}
    M&\coloneqq\mathbb C[x^\mu,\theta^\alpha,\lambda^\alpha]/(\gamma^\mu_{\alpha\beta}\lambda^\alpha\lambda^\beta)\\
    &\cong\left(\bigodot(10000)_{\mathfrak o(V)}\right)\otimes\left(\bigwedge(00010)_{\mathfrak o(V)}\right)\otimes\left(\bigoplus_{i=0}^\infty(000i0)_{\mathfrak o(V)}\right)
\end{aligned}
\end{equation}
where the \(\mathfrak o(V)\)-representation has been specified by Dynkin labels; we use the index notation where \(V\) indices are \(_\mu\) and \(S_+\) indices are \(_\alpha\) (hence \(S_-\cong S_+^*\) indices are \(^\alpha\)). The formal variables \(x,\theta,\lambda\) transform as \(V,S_-,S_-\) respectively under \(\mathfrak o(V)\), which in turn determines the \(\mathfrak o(V)\)-representation on \(M\).
The generators carry the degrees\footnote{In fact, this grading can be refined into a bigrading \cite[(3.15)~ff.]{Eager_2022}, and the \(\mathfrak p\)-representation respects this bigrading. But we do not need this fact.}
\begin{align}\label{eq:total_grading}
    |x|&=-2 & |\theta| &= -1 & |\lambda|&= 0.
\end{align}
 The differential is \cite[(3.14), (3.19)]{Eager_2022}
\begin{equation}
    d \coloneqq \lambda^\alpha\left(\frac\partial{\partial\theta^\alpha}-\gamma^\mu_{\alpha\beta}\theta^\beta\frac\partial{\partial x^\mu}\right).
\end{equation}
The \(\mathfrak o(V)\)-representation of \(M\) extends to a strict representation of \(\mathfrak p\) as
\begin{subequations}
  \label{eq:prep}
\begin{align}
   \rho^{(1)}_0(\tte_\mu)&= \frac\partial{\partial x^\mu}, \\
   \rho^{(1)}_0(\ttd_\alpha) &= \frac\partial{\partial\theta^\alpha}+\gamma^\mu_{\alpha\beta}\theta^\beta\frac\partial{\partial x^\mu}. \label{eq:10d_supertranslation_action}
\end{align}
\end{subequations}

We can twist this \(\mathfrak p\)-representation to obtain a (strict) representation  of \(\mathfrak p_Q\) on \(M_Q\); the action of $\mf p_Q$ is through \eqref{eq:prep}, but the differential on $M_Q$ is now $d+\rho_{0}^{(1)}(Q)$. The cohomology of $M_Q$ is the ring of regular functions on \(\mathbb A^5\).
\begin{theorem}[{\cite[Theorem 3.A]{saberi2021twisting}}]\label{thm:saberi_theorem}
The cohomology of \(M_Q\) is
\begin{equation}\operatorname H(M_Q)=\mathbb C[z^i],\end{equation}
where \(z^i=(z^1,\dotsc,z^5)\) is a formal variable of degree \(-2\) transforming under \(\mathfrak{sl}(L)\) as the defining representation \(L\), such that the \(L_\infty\)-representation of the subalgebra \(\mathfrak{isl}(L)\) of \(\operatorname H(\mathfrak p_Q)\) is
\begin{align}
\rho^{(1)}(\tilde\ttr^i{}_j)&=z^i\frac\partial{\partial z^j}-\frac15\delta^i_jz^k\frac\partial{\partial z^k},&
\rho^{(1)}(\tte_i)&=\frac\partial{\partial z^i},
\end{align}
with \(\rho^{(k)}=0\) for \(k\ge2\).
\end{theorem}
\noindent This corresponds to the fact that the holomorphic twist of ten-dimensional supersymmetric Yang--Mills theory is holomorphic Chern--Simons theory \cite{Baulieu:2010ch,elliott2020taxonomy,saberi2021twisting}, whose space of fields is\footnote{up to issues such as holomorphic versus algebraic functions, which we ignore} the algebraic Dolbeault complex of \(\mathbb A^5\), namely \(\mathbb C[z^i,\bar z_i,\mathrm d\bar z^i]\), and whose cohomology is therefore \(\mathbb C[z^i]\).

Furthermore, the cohomology \(\mathbb C[z^i]=\operatorname H(M_Q)\) is included into \(M_Q\) as a \(\mathfrak{sl}(L)\)-subrepresentation \cite{saberi2021twisting}. Since \(\mathfrak{sl}(L)\) is simple, there exists an $\mf{sl}(L)$-equivariant strong deformation retract of cochain complexes
  \begin{equation}\label{eq:module_sdr}
    \begin{tikzcd}
        \ar[loop left,"{(h,h')}"] \mf p_Q\oplus M_Q \rar[shift left, "{(p,p')}"] &  \lar[shift left, "{(i,i')}"] \operatorname H(\mathfrak p_Q)\oplus\operatorname H(M_Q),
    \end{tikzcd}
\end{equation}
whose restriction to \(\mathfrak p_Q\leftrightarrow\operatorname H(\mathfrak p_Q)\) is the strong deformation-retract \((i,p,h)\) given in \eqref{eq:ogretract}.
Thus, by taking the minimal model of \((\mathfrak p_Q,M_Q)\) along the strong deformation retract \eqref{eq:module_sdr}, the above \(\mathfrak{isl}(L)\)-representation extends into an \(L_\infty\)-representation of the entirety of \(\operatorname H(\mathfrak p_Q)\), and such minimal models are unique up to quasi-isomorphisms of \(L_\infty\)-algebra representations.
This minimal model is an $L_\infty$-representation of the $L_\infty$-algebra $\operatorname H(\mf p_Q)$ on $\operatorname H(M_Q)=\mathbb C[z^i]$.
The following theorem computes this minimal model explicitly.
\begin{theorem}\label{thm:H(p_Q)-action-on-A5}
The minimal model of the \(L_\infty\)-representation of \(\mathfrak p_Q\) on \(M_Q\) obtained using the strong deformation retract \((i,i';p,p';h,h')\) is the \(\operatorname H(\mathfrak p_Q)\)-representation on \(\operatorname H(M_Q)=\mathbb C[z^i]\) given by
\begin{equation}
\begin{aligned}
\rho^{(1)}(\tilde\ttr^i{}_j) &= z^i \frac\partial{\partial z^j}-\frac15\deltaup^i_jz^k\frac\partial{\partial z^k}, \\
\rho^{(1)}(\tte_i) &= \frac\partial{\partial z^i}, \\
\rho^{(3)}(\ttr^{ij},\ttr^{kl},\ttr^{mn}) &=
\frac12\Bigg(z^{[i}\varepsilon^{j]klp[m}z^{n]}
-
z^{[i}\varepsilon^{j]mnp[k}z^{l]}\\
&\qquad+z^{[k}\varepsilon^{l]mnp[i}z^{j]}
-
z^{[k}\varepsilon^{l]ijp[m}z^{n]}\\
&\qquad+z^{[m}\varepsilon^{n]ijp[k}z^{l]}
-
z^{[m}\varepsilon^{n]klp[i}z^{j]}\Bigg)\frac\partial{\partial z^p}\\
&\eqqcolon P^{ijklmn;}{}_{pq}^rz^pz^q\frac\partial{\partial z^r},
\end{aligned}
\end{equation}
with all other components vanishing (in particular, \(\rho^{(1)}(\ttr^{ij})=0\)),
where \(P^{ijklmn;}{}_{pq}^r\) is the projection
\begin{equation}
  \label{eq:rho3proj}
    (0100)_{\mathfrak{sl}(L)}^{\wedge3}=(0020)_{\mathfrak{sl}(L)}\oplus(2001)_{\mathfrak{sl}(L)}\to(2001)_{\mathfrak{sl}(L)}
\end{equation}
in terms of \(\mathfrak{sl}(L)\) Dynkin labels
or, in Young tableau notation,
\begin{equation}
    \ydiagram{1,1}^{\wedge3}
    =
    \ydiagram{2,2,2}
    \oplus
    \ydiagram{3,1,1,1}
    \to
    \ydiagram{3,1,1,1}.
\end{equation}
\end{theorem}\begin{proof}
First, note that for degree reasons, we can only have nonzero \(\rho^{(k)}\) for odd \(k\) since \(\mathbb C[z^i]\) and \(\operatorname H(\mathfrak p_Q)\) are all concentrated in even degree and \(\rho^{(k)}\) carries degree \(1-k\). The leading component \(\rho^{(1)}\) is fixed simply by restriction of the \(\mathfrak p_Q\)-representation \(\rho_0\) on \(M_Q\) to \(i(\operatorname H(\mathfrak p_Q))\subset\mathfrak p_Q\) and \(i'(\operatorname H(M_Q))\subset M_Q\) as 
\begin{equation}
    \rho^{(1)}(x) = p'\circ\rho_0(i(x))\circ i'
\end{equation}
for \(x\in\operatorname H(\mathfrak p_Q)\) (so that \(i(x)\in\mathfrak p_Q\));
in particular, \(\rho^{(1)}(\ttr^{ij})=0\). Furthermore, the \(\rho^{(k)}\) vanish whenever one of the arguments is \(\tilde\ttr^i{}_j\) except when \(k=1\) (\cref{lem:rule_out_sl(5)}).

Hence, it suffices to determine \(\rho^{(3)},\rho^{(5)},\rho^{(7)},\dotsc\) where all arguments are either \(\ttr^{ij}\) or \(\tte_i\). Now, the possibilities of \(\rho^{(k)}\) are constrained by the fact that all operations \(\mu_k\), \(\rho^{(k)}\), and the strong deformation retract \((i,i';p,p';h,h')\) are \(\mathfrak{sl}(L)\)-equivariant. Suppose that \(\rho^{(k)}\) does not vanish when fed \(p\) arguments of the form \(\ttr^{ij}\) and \(q\) arguments of the form \(\tte_i\) with \(p+q=k\equiv1\pmod2\). Then, representation-theoretically, it must yield a nontrivial \(\mathfrak{sl}(L)\)-representation that is a direct summand of
\begin{equation}\label{eq:fiddledee}
    \left( L^{\wedge2}\right)^{\wedge p}\otimes(L^*)^{\wedge q}.
\end{equation}
On the other hand, it must carry the degree \((1-p-q)+2q=1-p+q\), and hence be a sum of terms of the form
\begin{equation}
    \begin{cases}
        z^n\left(\frac\partial{\partial z}\right)^{n+(1-p+q)/2}&\text{if \(1-p+q\ge0\)}\\
        z^{n-(1-p+q)/2}\left(\frac\partial{\partial z}\right)^n&\text{if \(1-p+q\le0\)},
    \end{cases}
\end{equation}
where $z^n$ refers to a product $z^{i_1}\cdots z^{i_n}$, and similarly for $(\pdv{}{z})^n$. Since \(z^i\) transforms as \(L\) and \(\partial/\partial z^i\) as \(L^*\), this must transform under \(\mathfrak{sl}(L)\) as a direct summand of
\begin{equation}\label{eq:fiddledum}
    \begin{cases}
        L^{\odot n}\otimes(L^*)^{\odot(n+(1-p+q)/2)}&\text{if \(1-p+q\ge0\)}\\
        L^{\odot(n-(1-p+q)/2)}\otimes(L^*)^{\odot n}&\text{if \(1-p+q\le0\)}.
    \end{cases}
\end{equation}
Thus, the two \(\mathfrak{sl}(L)\)-representations \eqref{eq:fiddledee} and \eqref{eq:fiddledum} must share some nontrivial subrepresentations if the corresponding \(\rho^{(k)}\) is to not vanish. \Cref{lem:representation_constraint_on_rho} shows that this is only possible for \((p,q)=(3,0)\) and \((p,q)=(4,3)\), corresponding to
\begin{equation}\label{eq:rho(3)_ansatz}
\rho^{(3)}(\ttr^{ij},\ttr^{kl},\ttr^{mn})= P^{ijklmn;}{}_{pq}^r\left(\alpha_0z^pz^q\frac\partial{\partial z^r}+\alpha_1z^pz^qz^s\frac{\partial^2}{\partial z^r\partial z^s}+\dotsb\right)
\end{equation}
and
\begin{multline}
\rho^{(7)}(\ttr^{ij},\ttr^{kl},\ttr^{mn},\ttr^{pq},\tte_r,\tte_s,\tte_t)\\=P^{ijklmnpq;}_{rst}{}^u_v\left(\beta_0z^v\frac\partial{\partial z^u}
    +
    \beta_1z^vz^w\frac{\partial^2}{\partial z^u\partial z^w}+\dotsb\right),\label{eq:rho(7)_ansatz}
\end{multline}
respectively, where \(P^{ijklmn;}{}_{pq}^r\) is the projector \((0100)^{\wedge3}_{\mathfrak{sl}(L)}\to(2001)_{\mathfrak{sl}(L)}\) as in \eqref{eq:rho3proj} and \(P^{ijklmnpq;}_{rst}{}^u_v\) is the projector \((0100)^{\wedge4}_{\mathfrak{sl}(L)}\otimes(0001)^{\wedge3}_{\mathfrak{sl}(L)}\to(1001)_{\mathfrak{sl}(L)}\).

Now, we must solve the coherence relations. One \(L_\infty\)-representation coherence relation states\footnote{In this coherence relation, in our case, terms of the form \([\rho^{(0)},\rho^{(4)}]\), \([\rho^{(2)},\rho^{(2)}]\), \(\rho^{(4)}(\mu_1)\), \(\rho^{(3)}(\mu_2)\), and \(\rho^{(2)}(\mu_3)\) vanish.}
\begin{equation}\label{eq:rho3_coherence_relation}
    [\rho^{(1)}(\tte_i),\rho^{(3)}(\ttr^{jk},\ttr^{lm},\ttr^{np})]=\rho^{(1)}(\mu_4(\tte_i,\ttr^{jk},\ttr^{lm},\ttr^{np})).
  \end{equation}
  Substituting \(\rho^{(1)}(\tte_i)=\partial/\partial z^i\) and \eqref{eq:rho(3)_ansatz} into \eqref{eq:rho3_coherence_relation} yields
  \begin{equation}
    \left[
        \frac\partial{\partial z^i},
        P^{jklmnp;}{}_{qr}^s\left(\alpha_0z^qz^r\frac\partial{\partial z^s}+\alpha_1z^qz^rz^t\frac{\partial^2}{\partial z^s\partial z^t}+\dotsb\right)
    \right]\\
    =
    P^{jklmnp;}_i{}^r_qz^q\frac\partial{\partial z^r},
  \end{equation}
  where \(P^{jklmnp;}_i{}^r_q\) is the projector \eqref{eq:mu4_projector}.
  Solving this yields \(\alpha_0=1\) and \(\alpha_n=0\) for \(n>0\).
  
  Next, we have the \(L_\infty\)-module coherence relation\footnote{In this coherence relation, terms of the form \([\rho^{(0)},\rho^{(8)}]\), \([\rho^{(2)},\rho^{(6)}]\), \([\rho^{(3)},\rho^{(5)}]\), \([\rho^{(4)},\rho^{(3)}]\), \(\rho^{(1)}(\mu_8)\), \(\rho^{(2)}(\mu_7)\), \dots, \(\rho^{(8)}(\mu_1)\) vanish.}
\begin{equation}\label{eq:rho(7)_coherence}
    0=\left[\rho^{(7)}(\ttr^{ij},\ttr^{kl},\ttr^{mn},\ttr^{pq},\tte_{[r},\tte_s,\tte_t),\rho^{(1)}(\tte_{u]})\right].
\end{equation}
Plugging in the ansatz \eqref{eq:rho(7)_ansatz} into \eqref{eq:rho(7)_coherence} yields
\begin{equation}
    0=\left[P^{ijklmnpq;}_{[rst|}{}^v_w\left(\beta_0z^w\frac\partial{\partial z^v}
    +
    \beta_1z^wz^x\frac{\partial^2}{\partial z^v\partial z^x}+\dotsb\right),\frac\partial{\partial z^{|u]}}\right].
\end{equation}
Solving this shows that the coefficients \(\beta_0,\beta_1,\dotsc\) must all vanish since non-constant-coefficient differential operators do not commute with \(\partial/\partial z^l\).
\end{proof}

\begin{lemma}\label{lem:rule_out_sl(5)}
In the \(L_\infty\)-representation of the \(L_\infty\)-algebra \(\operatorname H(\mathfrak p_Q)\) on \(\mathbb C[z^i]\) obtained by homotopy transfer from \(M_Q\) using the \(\mathfrak{sl}(L)\)-equivariant strong deformation retract \((i,i';p,p';h,h')\), we have
\begin{equation}
    \rho^{(k)}(\tilde\ttr^i{}_j,\dotsc) = 0
\end{equation}
if \(k\ge2\).
\end{lemma}\begin{proof}
We are to perform a homotopy transfer of $L_\infty$-algebras along
\begin{equation}
    \begin{tikzcd}
        \ar[loop left,"{(h,h')}"] \mf p_Q\oplus M_Q \rar[shift left, "{(p,p')}"] &  \lar[shift left, "{(i,i')}"] \operatorname H(\mathfrak p_Q)\oplus\operatorname H(M_Q).
    \end{tikzcd}
  \end{equation}
Let us again use Feynman-diagrammatic terminology to refer to elements as `states'.
If \(\rho^{(k)}(\tilde\ttr^i{}_j,\dotsc)\ne0\), this would mean that there is at least one tree Feynman diagram with at least one external leg corresponding to \(\tilde\ttr^i{}_j\). Assuming that \(k\ge2\), we have the following possibilities.
\begin{enumerate}[(i)]
\item\label{case1} The vertex connected to this leg may be directly connected to \(p'\) as
\begin{equation}\label{eq:case_immediate_output}
    p'(\rho_0(\tilde\ttr^i{}_j)X)
    =
  \begin{tikzpicture}[scale=0.5,baseline={([yshift=-1ex]current bounding box.center)}]
    \draw [thick] (-2,0) -- (0,-2);
    \draw [thick, dashed] (2,0) -- (0,-2);
    \draw [thick, dashed] (0,-2) -- (0,-3);
    \node at (0,-3.5) {\footnotesize\(p'\)};
    \node at (2.5,0.5) {\footnotesize\(X\)};
    \node at (-2,0.5) {\footnotesize\(\tilde\ttr^i{}_j\)};
  \end{tikzpicture},
\end{equation}
where \(X\in M_Q\). In this case, we may assume \(X\) to be an intermediate state \(X=h'(\tilde X)\). (The alternative, that \(X\) lies in the cohomology, only yields \(\rho^{(1)}\).)
\item\label{case2} The vertex connected to this leg may feed into \(h\) and connect to the rest of the tree as
\begin{equation}\label{eq:case_not_immediate_output}
    p'(\dotsb h'(\rho_0(\tilde\ttr^i{}_j)X)\dotsb)
    =
  \begin{tikzpicture}[scale=0.5,baseline={([yshift=-1ex]current bounding box.center)}]
    \draw [thick] (-2,0) -- (0,-2);
    \draw [thick, dashed] (2,0) -- (0,-2);
    \draw [thick,dashed] (0,-2) -- (0,-5);
    \filldraw[color=white, fill=white] (0,-3.5) circle (0.5);
    \node at (0,-3.5) {\footnotesize\(h'\)};
    \node at (0,-5.5) {\footnotesize\(\vdots\)};
    \node at (2.5,0.5) {\footnotesize\( X\)};
    \node at (-2,0.5) {\footnotesize\(\tilde\ttr^i{}_j\)};
  \end{tikzpicture},
\end{equation}
where \(X\in M_Q\) may be either an intermediate state \(X=h'(\tilde X)\) or belong to the cohomology (\(X\in i'(\operatorname H(M_Q))\)). In either case, we have \(h(X)=0\).
\item\label{case3} The vertex connected to this leg may feed into \(h\) and connect to the rest of the tree as
\begin{equation}\label{eq:case_not_immediate_output_algebra}
    p'(\dotsb h[\tilde\ttr^i{}_j,x]\dotsb)
    =
  \begin{tikzpicture}[scale=0.5,baseline={([yshift=-1ex]current bounding box.center)}]
    \draw [thick] (-2,0) -- (0,-2);
    \draw [thick] (2,0) -- (0,-2);
    \draw [thick] (0,-2) -- (0,-5);
    \filldraw[color=white, fill=white] (0,-3.5) circle (0.5);
    \node at (0,-3.5) {\footnotesize\(h\)};
    \node at (0,-5.5) {\footnotesize\(\vdots\)};
    \node at (2.5,0.5) {\footnotesize\(x\)};
    \node at (-2,0.5) {\footnotesize\(\tilde\ttr^i{}_j\)};
  \end{tikzpicture},
\end{equation}
where \(x\in\mathfrak p_Q\) is either an intermediate state \(x=h(\tilde x)\) or belongs to the cohomology (\(x\in i(\operatorname H(\mathfrak p_Q))\)). In either case, \(h(x)=0\).
\end{enumerate}

In the first case \eqref{eq:case_immediate_output}, since the strong deformation retract \((i,i';p,p';h,h')\) is \(\mathfrak{sl}(L)\)-equivariant, \(p'(\rho_0(\tilde\ttr^i{}_j)h'(\tilde X))\) can be nonzero only if \(p'(h'(\tilde X))\) is already nonzero. But this cannot be the case since \((i',p',h')\) forms a strong deformation retract, whose definition requires \(p'\circ h'=0\).

Similarly, in the latter case \eqref{eq:case_not_immediate_output},
since the strong deformation retract \((i,i';p,p';h,h')\) is \(\mathfrak{sl}(L)\)-equivariant, \(h'(\rho_0(\tilde\ttr^i{}_j)X)\) can be nonzero only if \(h'(X))\) is already nonzero, but this is not possible.

Finally, in the last case \eqref{eq:case_not_immediate_output_algebra}, again, since the strong deformation retract \((i,i';p,p';h,h')\) is \(\mathfrak{sl}(L)\)-equivariant, \(h[\tilde\ttr^i{}_j,x]\) can be nonzero only if \(h(x)\) is already nonzero, which is not possible.
\end{proof}

\begin{lemma}\label{lem:representation_constraint_on_rho}
For \(p+q\) odd and \(p+q\ge3\), the \(\mathfrak{sl}(L)\)-representation
\begin{equation}
    R_{p,q}\coloneqq(L^{\wedge2})^{\wedge p}\otimes(L^*)^{\wedge q}
\end{equation}
has no irreducible components in common with
\begin{equation}
    \tilde R_{p,q}\coloneqq\begin{cases}
        \bigoplus_{n=0}^\infty L^{\odot n}\otimes(L^*)^{\odot(n+(1+q-p)/2)} &\text{if } 1+q-p\ge0 \\
        \bigoplus_{n=0}^\infty L^{\odot(n-(1+q-p)/2)}\otimes(L^*)^{\odot n} &\text{if } 1+q-p\le0
    \end{cases}
\end{equation}
except when \((p,q)=(3,0)\) or \((4,3)\), in which case the irreducible components in common are
\((2001)_{\mathfrak{sl}(L)}\) and \((1001)_{\mathfrak{sl}(L)}\)
respectively.
\end{lemma}\begin{proof}
We must compute the tensor product appearing in \(\tilde R_{p,q}\). For \(1+q-p\ge0\) and any nonnegative integer \(n\ge0\), it is easy to see that
\begin{equation}
   \big(n000\big)\otimes\big(000(n+(1+q-p)/2)\big)
    =\bigoplus_{i=0}^n\big(i00(i+(1+q-p)/2)\big).
  \end{equation}
  Similarly, for \(1+q-p\le0\) we have
  \begin{equation}
   \big( (n-(1+q-p)/2))000\big)\otimes\big(000n\big)
    =\bigoplus_{i=0}^n\big((i-(1+q-p)/2))00i\big).
  \end{equation}
Thus, \(\tilde R_{p,q}\) only contains irreducible representations of the form
\begin{equation}
    \begin{cases}
    \big(i00(i+(1+q-p)/2)\big)&\text{if \(1+q-p\ge0\)}\\
   \big( (i-(1+q-p)/2)00i\big)&\text{if \(1+q-p\le0\)}
    \end{cases}
    \qquad(i\in\{0,1,2,\dotsc\}).
\end{equation}
Given this, iterating over\footnote{Recall that $(L^{\wedge 2})^{\wedge k}$ and $(L^*)^{\wedge l}$ are zero for $k>10$ and $l>5$.} \(p\in\{0,1,\dotsc,10\}\) and \(q\in\{0,1,\dotsc,5\}\) and verifying whether an irreducible representation of the above form appears in \(R_{p,q}\) (using e.g.~a computer algebra system), one can see that \((p,q)\in\{(3,0),(4,3)\}\) are the only possible solutions. 
\end{proof}

\section{Other dimensions and amounts of supersymmetry}\label{sec:higher_d}
The above construction works for general supersymmetry algebras and general supermultiplets, but ten-dimensional \(\mathcal N=(1,0)\) super-Poincaré algebra seems to be one of the very few in having nontrivial and purely bosonic higher products, at least if one is to start from the vector supermultiplet; a glance at \cite{elliott2020taxonomy} shows that this is the only case in which the cohomology is simply a polynomial ring in bosonic variables.

In general, for sufficiently large dimension \(n\), the number of spinorial components in a super-Poincaré algebra increases as \(\mathcal O(2^n)\) whereas bosonic components increase as \(\mathcal O(n^2)\). Indeed, already at 14 dimensions, the minimal spinor has 128 components while \(\mathfrak{io}(14)\) has 105 components. So we are restricted to 13 or fewer dimensions (unless we twist again to eliminate more supertranslations) even if one did not take into consideration no-go theorems about higher-spin theories (since we ignore dynamics here). Similarly, the 11-dimensional case (starting with the supergravity multiplet) is discussed in \cite{saberi2021twisting,Hahner:2023kts}. There are two possible cases. In one case \cite{saberi2021twisting}, \(\operatorname H(\mathfrak p_Q)\) contains fermionic elements. Then we expect the action of \(\operatorname H(\mathfrak p_Q)\) to contain a \(\rho^{(2)}\) involving the remaining fermionic elements.
In the other case \cite{Hahner:2023kts}, however, we expect to see a higher action of
\begin{equation}
    \operatorname H(\mathfrak p_Q) = (\mathfrak g_2\oplus\mathfrak{sl}(L)\oplus V_7\otimes L\oplus\mathbb C)\ltimes L
\end{equation}
(which should carry nontrivial \(\mu_2\) and \(\mu_4\))
on
\begin{equation}\mathbb A(L^*)=\operatorname{Spec}\mathbb C[z_1,z_2],\end{equation}
where \(L\) is a two-dimensional vector space and \(V_7\) is a seven-dimensional vector space; here, \(\mathbb C[z_1,z_2]=\operatorname H(\Omega^{0,\bullet}(L)\otimes\Omegaup^\bullet(V))\) is the cohomology of the Dolbeault--de~Rham complex on two complex and seven real dimensions.

On the other hand, if there are too few dimensions, higher products may vanish. For example, for the four-dimensional \(\mathcal N=1\) super-Poincaré algebra, the twist gives a decomposition of four-dimensional complexified spacetime \(V\) as \(V=L\oplus L^*\), where \(L\) is a two-dimensional vector space, with the twisted super-Poincaré algebra being \cite{saberi2021twisting}
\begin{equation}
    \mathfrak p_Q=\left(\begin{tikzcd}
         (L^*)^{\wedge2} \rar & L^{\wedge0} \\
        (\mathfrak{sl}(L)\ltimes L^{\wedge2})\oplus\mathfrak{gl}(1)_{\mathrm R} & L^{\wedge1}\rar & L \\
        \mathfrak{gl}(1)_{\mathrm{tr}} \rar &  L^{\wedge2} & L^*
    \end{tikzcd}\right),
\end{equation}
where \(\mathfrak{gl}(1)_{\mathrm R}\) is the R-symmetry and \(\mathfrak{gl}(1)_{\mathrm{tr}}\) is the trace part of \(\mathfrak{gl}(L)\).
Following the proof of \cref{thm:H(p_Q)}, we see that there are no higher brackets for \(\operatorname H(\mathfrak p_Q)\) by constructing all possible intermediate states: apply \(\mu_2( L^{\wedge2},-)\) to \(L^*[-2]\) to get \(L[-2]\); applying the homotopy \(h\) yields \(L^{\wedge1}\); but now applying another \(\mu_2( L^{\wedge2},-)\) simply kills everything.

\section*{Acknowledgements}
The authors (\acctextsc{DSHJ}, \acctextsc{HK}, and \acctextsc{CASY}) were supported in part by the Leverhulme Research Project Grant \acctextsc{RPG}--2021--092. The authors thank Fridrich Valach~III\orcidlink{0000-0003-0020-1999}, Charles Strickland-Constable\orcidlink{0000-0003-0294-1253}, Pieter Bomans\orcidlink{0000-0002-0907-9830}, and Jingxiang Wu\orcidlink{0000-0001-6867-1407} for helpful discussions; \acctextsc{casy} thanks Brian R. Williams\orcidlink{0000-0001-9641-861X}, and \acctextsc{dshj} thanks Ingmar Akira Saberi\orcidlink{0000-0002-2005-938X}.
\appendix
\section{Conventions}\label{app:conventions} 
In a basis adapted to the choice of pure spinor $Q\in S_+$, the super-Poincaré algebra $\mf p = \mf p^0 \oplus \mf p^1 \oplus \mf p^2$ in ten dimensions has the basis elements
\begin{align}
\ttr^i{}_j, \ttr^{ij}, \ttr_{ij} &\in \mf p^0,&
\ttd, \ttd_{ij},\ttd^i&\in \mf p^1,&
\tte^i,\tte_i&\in \mf p^2,
\end{align}
with \(\ttr^{ij}=-\ttr^{ji}\), \(\ttr_{ij}=-\ttr_{ji}\), and \(\ttd_{ij}=-\ttd_{ji}\).
The graded-skew-symmetric Lie brackets among these basis elements are
\begin{gather}
\begin{aligned}
\left[ \ttr^i{}_j, \ttr^k{}_l \right] &= \delta^k_j \ttr^i{}_l -\delta^i_l \ttr^k{}_j & 
\left[ \ttr^i{}_j, \ttr^{kl} \right] &= \delta^k_j \ttr^{il} + \delta^l_j \ttr^{ik} \\
\left[ \ttr^i{}_j, \ttr_{kl} \right] &= - \delta^i_k \ttr_{jl} - \delta^i_l \ttr_{kj} &
\left[ \ttr^{ij}, \ttr_{kl} \right] &= \delta^{i}_{k} \ttr^{j}{}_{l} - \delta^{j}_{k} \ttr^{i}{}_{l} - \delta^{i}_{l} \ttr^{j}{}_{k} + \delta^{j}_{l} \ttr^{i}{}_{k} \\
\left[ \ttr^i{}_j, \tte^k \right] &=  \delta^k_j \tte^i & 
\left[ \ttr^i{}_j, \tte_k \right] &= -\delta^i_k \tte_j \\
\left[ \ttr^{ij}, \tte_k \right] &=  \delta^j_k \tte^{i} - \delta^i_k \tte^{j} & 
\left[ \ttr_{ij}, \tte^k \right] &=  \delta_i^k \tte_{j} - \delta_j^k \tte_{i}  \\
\left[ \ttr^{ij} , \ttd \right] &= 0 & 
\left[ \ttr_{ij}, \ttd \right] &= -\ttd_{ij} \\  
\left[ \ttr^{ij} , \ttd_{kl} \right] &= (\delta^i_l \delta^j_k-\delta^i_k \delta^j_l) \ttd & 
\left[ \ttr_{ij},  \ttd_{kl} \right] &= -\varepsilon_{ijklm} \ttd^m \\  
\left[ \ttr^{ij} , \ttd^k \right] &= -\frac12\varepsilon^{ijklm} \ttd_{lm} & 
\left[ \ttr_{ij},  \ttd^k \right] &= 0 \\
\left[ \ttr^i{}_j, \ttd \right] &= \phantom + \frac 12 \delta^i_j \ttd &
 \left[ \ttr^i{}_j, \ttd_{kl} \right] &= \phantom + \frac 12 \delta^i_j \ttd_{kl} - \delta^i_k \ttd_{jl} - \delta^i_l \ttd_{kj} \\
\left[ \ttr^i{}_j, \ttd^k \right] &= -\frac 12 \delta^i_j \ttd^k + \delta^k_j \ttd^i &
\left[ \ttd , \ttd^i \right] &= \tte^i \\
\left[ \ttd^i , \ttd_{jk} \right] &= \delta^i_k \tte_j -\delta^i_j \tte_k  &
\left[ \ttd_{ij} , \ttd_{kl} \right] &= - \varepsilon_{ijklm} \tte^m 
\end{aligned}
\end{gather}
with all remaining brackets of basis elements vanishing. Here the indices $i,j,\dots$ run over $\{1,2,3,4,5\}$, and we employ the Einstein summation convention.
Here $\ttr^i{}_j$ span $\mf{gl}(5)$. The basis elements of the subalgebras $\mf{sl}(5)$ are
\begin{equation}
    \tilde\ttr^i{}_j\coloneqq\ttr^i{}_j - \frac15\delta^i_j\ttr^k{}_k.
\end{equation}

\bibliographystyle{unsrturl}
\bibliography{biblio}

\begin{thebibliography}{10}

\bibitem{Witten:1988ze}
Edward Witten.
\newblock Topological quantum field theory.
\newblock {\em Communications in Mathematical Physics}, 117:353--386, September
  1988.
\newblock \href {https://doi.org/10.1007/BF01223371}
  {\path{doi:10.1007/BF01223371}}.

\bibitem{2011arXiv1111.4234C}
Kevin~Joseph Costello.
\newblock Notes on supersymmetric and holomorphic field theories in dimensions
  2 and 4.
\newblock {\em Pure and Applied Mathematics Quarterly}, 9(1):73--165, 2013.
\newblock \href {https://arxiv.org/abs/1111.4234} {\path{arXiv:1111.4234}},
  \href {https://doi.org/10.4310/PAMQ.2013.v9.n1.a3}
  {\path{doi:10.4310/PAMQ.2013.v9.n1.a3}}.

\bibitem{Elliott:2018cbx}
Christopher~J. Elliott and Pavel Safronov.
\newblock Topological twists of supersymmetric algebras of observables.
\newblock {\em Communications in Mathematical Physics}, 371(2):727--786, March
  2019.
\newblock \href {https://arxiv.org/abs/1805.10806} {\path{arXiv:1805.10806}},
  \href {https://doi.org/10.1007/s00220-019-03393-9}
  {\path{doi:10.1007/s00220-019-03393-9}}.

\bibitem{elliott2020taxonomy}
Christopher~J. Elliott, Pavel Safronov, and Brian~R. Williams.
\newblock A taxonomy of twists of supersymmetric {Y}ang--{M}ills theory.
\newblock {\em Selecta Mathematica, New Series}, 28:73, August 2022.
\newblock \href {https://arxiv.org/abs/2002.10517} {\path{arXiv:2002.10517}},
  \href {https://doi.org/10.1007/s00029-022-00786-y}
  {\path{doi:10.1007/s00029-022-00786-y}}.

\bibitem{Howe:1991bx}
Paul~S. Howe.
\newblock Pure spinors, function superspaces and supergravity theories in ten
  and eleven dimensions.
\newblock {\em Physics Letters B}, 273(1--2):90--94, December 1991.
\newblock \href {https://doi.org/10.1016/0370-2693(91)90558-8}
  {\path{doi:10.1016/0370-2693(91)90558-8}}.

\bibitem{Berkovits:2002uc}
Nathan~Jacob Berkovits.
\newblock Towards covariant quantization of the supermembrane.
\newblock {\em Journal of High Energy Physics}, 2002(09):051, September 2002.
\newblock \href {https://arxiv.org/abs/hep-th/0201151}
  {\path{arXiv:hep-th/0201151}}, \href
  {https://doi.org/10.1088/1126-6708/2002/09/051}
  {\path{doi:10.1088/1126-6708/2002/09/051}}.

\bibitem{Eager_2022}
Richard Eager, Fabian Hahner, Ingmar~Akira Saberi, and Brian~R. Williams.
\newblock Perspectives on the pure spinor superfield formalism.
\newblock {\em Journal of Geometry and Physics}, 180:104626, October 2022.
\newblock \href {https://arxiv.org/abs/2111.01162} {\path{arXiv:2111.01162}},
  \href {https://doi.org/10.1016/j.geomphys.2022.104626}
  {\path{doi:10.1016/j.geomphys.2022.104626}}.

\bibitem{Elliott_2023}
Christopher~J. Elliott, Fabian Hahner, and Ingmar~Akira Saberi.
\newblock The derived pure spinor formalism as an equivalence of categories.
\newblock {\em Symmetry, Integrability and Geometry: Methods and Applications},
  19:022, 2023.
\newblock \href {https://arxiv.org/abs/2205.14133} {\path{arXiv:2205.14133}},
  \href {https://doi.org/10.3842/SIGMA.2023.022}
  {\path{doi:10.3842/SIGMA.2023.022}}.

\bibitem{jonssonthesis}
David Simon~Henrik Jonsson.
\newblock Supermultiplets and {K}oszul duality: Super-{Y}ang-{M}ills and
  supergravity using pure spinors, May 2021.
\newblock Master's thesis, Chalmers Tekniska Högskola.
\newblock URL: \url{https://hdl.handle.net/20.500.12380/302660}.

\bibitem{cederwall2024canonical}
Nils Martin~Sten Cederwall, David Simon~Henrik Jonsson, Lars~Jakob Palmkvist,
  and Ingmar~Akira Saberi.
\newblock Canonical supermultiplets and their {K}oszul duals.
\newblock {\em Communications in Mathematical Physics}, 405:127, May 2024.
\newblock \href {https://arxiv.org/abs/2304.01258} {\path{arXiv:2304.01258}},
  \href {https://doi.org/10.1007/s00220-024-04990-z}
  {\path{doi:10.1007/s00220-024-04990-z}}.

\bibitem{Cederwall:2013vba}
Nils Martin~Sten Cederwall.
\newblock Pure spinor superfields -- an overview.
\newblock In Stefano Bellucci, editor, {\em Breaking of Supersymmetry and
  Ultraviolet Divergences in Extended Supergravity. Proceedings of the
  INFN-Laboratori Nazionali di Frascati School 2013}, volume 153 of {\em
  Springer Proceedings in Physics}, pages 61--93, 2014.
\newblock \href {https://arxiv.org/abs/1307.1762} {\path{arXiv:1307.1762}},
  \href {https://doi.org/10.1007/978-3-319-03774-5_4}
  {\path{doi:10.1007/978-3-319-03774-5_4}}.

\bibitem{Cederwall:2022fwu}
Nils Martin~Sten Cederwall.
\newblock Pure spinors in classical and quantum supergravity.
\newblock In Cosimo Bambi, Leonardo Modesto, and Ilya~Lvovich Shapiro, editors,
  {\em Handbook of Quantum Gravity}. Springer, 2023.
\newblock \href {https://arxiv.org/abs/2210.06141} {\path{arXiv:2210.06141}},
  \href {https://doi.org/10.1007/978-981-19-3079-9_47-1}
  {\path{doi:10.1007/978-981-19-3079-9_47-1}}.

\bibitem{Cederwall:2009ez}
Nils Martin~Sten Cederwall.
\newblock Towards a manifestly supersymmetric action for 11-dimensional
  supergravity.
\newblock {\em Journal of High Energy Physics}, 2010(01):117, January 2010.
\newblock \href {https://arxiv.org/abs/0912.1814} {\path{arXiv:0912.1814}},
  \href {https://doi.org/10.1007/JHEP01(2010)117}
  {\path{doi:10.1007/JHEP01(2010)117}}.

\bibitem{Cederwall:2010tn}
Nils Martin~Sten Cederwall.
\newblock \({D}=11\) supergravity with manifest supersymmetry.
\newblock {\em Modern Physics Letters A}, 25(38):3201--3212, December 2010.
\newblock \href {https://arxiv.org/abs/1001.0112} {\path{arXiv:1001.0112}},
  \href {https://doi.org/10.1142/S0217732310034407}
  {\path{doi:10.1142/S0217732310034407}}.

\bibitem{Berkovits:2001rb}
Nathan~Jacob Berkovits.
\newblock Covariant quantization of the superparticle using pure spinors.
\newblock {\em Journal of High Energy Physics}, 2001(09):016, September 2001.
\newblock \href {https://arxiv.org/abs/hep-th/0105050}
  {\path{arXiv:hep-th/0105050}}, \href
  {https://doi.org/10.1088/1126-6708/2001/09/016}
  {\path{doi:10.1088/1126-6708/2001/09/016}}.

\bibitem{Cederwall:2001bt}
Nils Martin~Sten Cederwall, Bengt Erik~Willy Nilsson, and Dimitrios Tsimpis.
\newblock The structure of maximally supersymmetric {Y}ang-{M}ills theory:
  Constraining higher order corrections.
\newblock {\em Journal of High Energy Physics}, 2001(06):034, July 2001.
\newblock \href {https://arxiv.org/abs/hep-th/0102009}
  {\path{arXiv:hep-th/0102009}}, \href
  {https://doi.org/10.1088/1126-6708/2001/06/034}
  {\path{doi:10.1088/1126-6708/2001/06/034}}.

\bibitem{Cederwall:2008vd}
Nils Martin~Sten Cederwall.
\newblock \({N}=8\) superfield formulation of the
  {B}agger-{L}ambert-{G}ustavsson model.
\newblock {\em Journal of High Energy Physics}, 2008(09):116, September 2008.
\newblock \href {https://arxiv.org/abs/0808.3242} {\path{arXiv:0808.3242}},
  \href {https://doi.org/10.1088/1126-6708/2008/09/116}
  {\path{doi:10.1088/1126-6708/2008/09/116}}.

\bibitem{Cederwall:2008xu}
Nils Martin~Sten Cederwall.
\newblock {Superfield actions for \({N}=8\) and \({N}=6\) conformal theories in
  three dimensions}.
\newblock {\em Journal of High Energy Physics}, 2008(10):070, October 2008.
\newblock \href {https://arxiv.org/abs/0809.0318} {\path{arXiv:0809.0318}},
  \href {https://doi.org/10.1088/1126-6708/2008/10/070}
  {\path{doi:10.1088/1126-6708/2008/10/070}}.

\bibitem{Cederwall:2009ay}
Nils Martin~Sten Cederwall.
\newblock Pure spinor superfields, with application to \({D}=3\) conformal
  models.
\newblock {\em Proceedings of the Estonian Academy of Sciences},
  59(4):280--289, 2010.
\newblock \href {https://arxiv.org/abs/0906.5490} {\path{arXiv:0906.5490}},
  \href {https://doi.org/10.3176/proc.2010.4.05}
  {\path{doi:10.3176/proc.2010.4.05}}.

\bibitem{saberi2021twisting}
Ingmar~Akira Saberi and Brian~R. Williams.
\newblock Twisting pure spinor superfields, with applications to supergravity.
\newblock {\em Pure and Applied Mathematics Quarterly}, 20(2):645--701, 2024.
\newblock \href {https://arxiv.org/abs/2106.15639} {\path{arXiv:2106.15639}},
  \href {https://doi.org/10.4310/PAMQ.2024.v20.n2.a2}
  {\path{doi:10.4310/PAMQ.2024.v20.n2.a2}}.

\bibitem{Hahner:2023kts}
Fabian Hahner and Ingmar~Akira Saberi.
\newblock Eleven-dimensional supergravity as a {C}alabi-{Y}au twofold, April
  2023.
\newblock \href {https://arxiv.org/abs/2304.12371} {\path{arXiv:2304.12371}},
  \href {https://doi.org/10.48550/arXiv.2304.12371}
  {\path{doi:10.48550/arXiv.2304.12371}}.

\bibitem{Baulieu:2010ch}
Laurent Baulieu.
\newblock \(\operatorname{SU}(5)\)-invariant decomposition of ten-dimensional
  {Y}ang--{M}ills supersymmetry.
\newblock {\em Physics Letters B}, 698(1):63--67, March 2011.
\newblock \href {https://arxiv.org/abs/1009.3893} {\path{arXiv:1009.3893}},
  \href {https://doi.org/10.1016/j.physletb.2010.12.044}
  {\path{doi:10.1016/j.physletb.2010.12.044}}.

\bibitem{pirsa_PIRSA:23070028}
Ingmar~Akira Saberi.
\newblock Research talk 13: Six- and eleven-dimensional theories via superspace
  torsion and {P}oisson brackets.
\newblock Talk at Strings 2023, Perimeter Institute, July 2023.
\newblock \href {https://doi.org/10.48660/23070028}
  {\path{doi:10.48660/23070028}}.

\bibitem{Jurco:2018sby}
Branislav Jurčo, Lorenzo Raspollini, Christian Sämann, and Martin Wolf.
\newblock \({L}_\infty\)-algebras of classical field theories and the
  {B}atalin--{V}ilkovisky formalism.
\newblock {\em Fortschritte der Physik}, 67(7):1900025, July 2019.
\newblock \href {https://arxiv.org/abs/1809.09899} {\path{arXiv:1809.09899}},
  \href {https://doi.org/10.1002/prop.201900025}
  {\path{doi:10.1002/prop.201900025}}.

\bibitem{Macrelli:2019afx}
Tommaso Macrelli, Christian Sämann, and Martin Wolf.
\newblock Scattering amplitude recursion relations in
  {B}atalin-{V}ilkovisky--quantizable theories.
\newblock {\em Physical Review D}, 100(4):045017, August 2019.
\newblock \href {https://arxiv.org/abs/1903.05713} {\path{arXiv:1903.05713}},
  \href {https://doi.org/10.1103/PhysRevD.100.045017}
  {\path{doi:10.1103/PhysRevD.100.045017}}.

\bibitem{Jurco:2019yfd}
Branislav Jur\v{c}o, Tommaso Macrelli, Christian S\"amann, and Martin Wolf.
\newblock Loop amplitudes and quantum homotopy algebras.
\newblock {\em Journal of High Energy Physics}, 2020(07):003, July 2020.
\newblock \href {https://arxiv.org/abs/1912.06695} {\path{arXiv:1912.06695}},
  \href {https://doi.org/10.1007/JHEP07(2020)003}
  {\path{doi:10.1007/JHEP07(2020)003}}.

\bibitem{Borsten:2021hua}
Leron Borsten, Hyungrok Kim, Branislav Jurčo, Tommaso Macrelli, Christian
  Sämann, and Martin Wolf.
\newblock Double copy from homotopy algebras.
\newblock {\em Fortschritte der Physik}, 69(8--9):2100075, September 2021.
\newblock \href {https://arxiv.org/abs/2102.11390} {\path{arXiv:2102.11390}},
  \href {https://doi.org/10.1002/prop.202100075}
  {\path{doi:10.1002/prop.202100075}}.

\bibitem{Kajiura:2004xu}
Hiroshige Kajiura and James~Dillon Stasheff.
\newblock Homotopy algebras inspired by classical open-closed string field
  theory.
\newblock {\em Communications in Mathematical Physics}, 263:553--581, March
  2006.
\newblock \href {https://arxiv.org/abs/math/0410291}
  {\path{arXiv:math/0410291}}, \href
  {https://doi.org/10.1007/s00220-006-1539-2}
  {\path{doi:10.1007/s00220-006-1539-2}}.

\bibitem{Kajiura:2005sn}
Hiroshige Kajiura and James~Dillon Stasheff.
\newblock Open-closed homotopy algebra in mathematical physics.
\newblock {\em Journal of Mathematical Physics}, 47(2):023506, February 2006.
\newblock \href {https://arxiv.org/abs/hep-th/0510118}
  {\path{arXiv:hep-th/0510118}}, \href {https://doi.org/10.1063/1.2171524}
  {\path{doi:10.1063/1.2171524}}.

\bibitem{Kajiura:2006mt}
Hiroshige Kajiura and James~Dillon Stasheff.
\newblock Homotopy algebra of open–closed strings.
\newblock {\em Geometry \& Topology Monographs}, 13:229--259, February 2008.
\newblock \href {https://arxiv.org/abs/hep-th/0606283}
  {\path{arXiv:hep-th/0606283}}, \href
  {https://doi.org/10.2140/gtm.2008.13.229}
  {\path{doi:10.2140/gtm.2008.13.229}}.

\bibitem{costello2015quantization}
Kevin~Joseph Costello and Si~Li.
\newblock Quantization of open-closed {BCOV} theory, {I}, May 2015.
\newblock \href {https://arxiv.org/abs/1505.06703} {\path{arXiv:1505.06703}},
  \href {https://doi.org/10.48550/arXiv.1505.06703}
  {\path{doi:10.48550/arXiv.1505.06703}}.

\bibitem{costello2020anomaly}
Kevin~Joseph Costello and Si~Li.
\newblock Anomaly cancellation in the topological string.
\newblock {\em Advances in Theoretical and Mathematical Physics},
  24(7):1723--1771, 2020.
\newblock \href {https://arxiv.org/abs/1905.09269} {\path{arXiv:1905.09269}},
  \href {https://doi.org/10.4310/ATMP.2020.v24.n7.a2}
  {\path{doi:10.4310/ATMP.2020.v24.n7.a2}}.

\bibitem{Beem:2013sza}
Christopher Beem, Madalena Lemos, Pedro Liendo, Wolfger Peelaers, Leonardo
  Rastelli, and Balt~C. van Rees.
\newblock Infinite chiral symmetry in four dimensions.
\newblock {\em Communications in Mathematical Physics}, 336(3):1359--1433,
  2015.
\newblock \href {https://arxiv.org/abs/1312.5344} {\path{arXiv:1312.5344}},
  \href {https://doi.org/10.1007/s00220-014-2272-x}
  {\path{doi:10.1007/s00220-014-2272-x}}.

\bibitem{Saberi:2019fkq}
Ingmar~Akira Saberi and Brian~R. Williams.
\newblock Superconformal algebras and holomorphic field theories.
\newblock {\em Annales Henri Poincaré}, 24(2):541--604, 2023.
\newblock \href {https://arxiv.org/abs/1910.04120} {\path{arXiv:1910.04120}},
  \href {https://doi.org/10.1007/s00023-022-01224-7}
  {\path{doi:10.1007/s00023-022-01224-7}}.

\bibitem{Bomans:2023mkd}
Pieter Bomans and Jingxiang Wu.
\newblock {Unravelling the Holomorphic Twist: Central Charges}.
\newblock 11 2023.
\newblock \href {https://arxiv.org/abs/2311.04304} {\path{arXiv:2311.04304}}.

\bibitem{Gaiotto:2024gii}
Davide Gaiotto, Justin Kulp, and Jingxiang Wu.
\newblock Higher operations in perturbation theory.
\newblock 3 2024.
\newblock \href {https://arxiv.org/abs/2403.13049} {\path{arXiv:2403.13049}}.

\bibitem{Loday:2012aa}
Jean-Louis Loday and Bruno Vallette.
\newblock {\em Algebraic operads}, volume 346 of {\em Grundlehren der
  mathematischen Wissenschaften}.
\newblock Springer, August 2012.
\newblock \href {https://doi.org/10.1007/978-3-642-30362-3}
  {\path{doi:10.1007/978-3-642-30362-3}}.

\bibitem{dotsenko2019twisting}
Vladimir~Viktorovich Dotsenko, Sergey~Viktorovich Shadrin, and Bruno Vallette.
\newblock The twisting procedure, October 2018.
\newblock \href {https://arxiv.org/abs/1810.02941} {\path{arXiv:1810.02941}},
  \href {https://doi.org/10.48550/arXiv.1810.02941}
  {\path{doi:10.48550/arXiv.1810.02941}}.

\bibitem{kraft2022introduction}
Andreas Kraft and Jonas Schnitzer.
\newblock An introduction to ${L}_\infty$-algebras and their homotopy theory
  for the working mathematician.
\newblock {\em Reviews in Mathematical Physics}, 36(1):2330006, February 2024.
\newblock \href {https://arxiv.org/abs/2207.01861} {\path{arXiv:2207.01861}},
  \href {https://doi.org/10.1142/S0129055X23300066}
  {\path{doi:10.1142/S0129055X23300066}}.

\bibitem{Dotsenko_2023}
Vladimir~Viktorovich Dotsenko, Sergey~Viktorovich Shadrin, and Bruno Vallette.
\newblock {\em {M}aurer--{C}artan Methods in Deformation Theory: the Twisting
  Procedure}, volume 488 of {\em London Mathematical Society Lecture Note
  Series}.
\newblock Cambridge University Press, August 2023.
\newblock \href {https://arxiv.org/abs/2212.11323} {\path{arXiv:2212.11323}},
  \href {https://doi.org/10.1017/9781108963800}
  {\path{doi:10.1017/9781108963800}}.

\bibitem{Berglund_2014}
Alexander Frans~Olof Berglund.
\newblock Homological perturbation theory for algebras over operads.
\newblock {\em Algebraic \&\ Geometric Topology}, 14(5):2511–2548, November
  2014.
\newblock \href {https://arxiv.org/abs/0909.3485} {\path{arXiv:0909.3485}},
  \href {https://doi.org/10.2140/agt.2014.14.2511}
  {\path{doi:10.2140/agt.2014.14.2511}}.

\bibitem{Saemann:2020oyz}
Christian Sämann and Emmanouil Sfinarolakis.
\newblock Symmetry factors of {F}eynman diagrams and the homological
  perturbation lemma.
\newblock {\em Journal of High Energy Physics}, 2020(12):088, December 2020.
\newblock \href {https://arxiv.org/abs/2009.12616} {\path{arXiv:2009.12616}},
  \href {https://doi.org/10.1007/JHEP12(2020)088}
  {\path{doi:10.1007/JHEP12(2020)088}}.

\bibitem{Lada:1994mn}
Thomas~Joseph Lada and Martin Markl.
\newblock Strongly homotopy {L}ie algebras.
\newblock {\em Communications in Algebra}, 23(6):2147--2161, 1995.
\newblock \href {https://arxiv.org/abs/hep-th/9406095}
  {\path{arXiv:hep-th/9406095}}, \href
  {https://doi.org/10.1080/00927879508825335}
  {\path{doi:10.1080/00927879508825335}}.

\bibitem{lada}
Thomas~Joseph Lada.
\newblock \({L}_\infty\)-algebra representations.
\newblock {\em Applied Categorical Structures}, 12:29--34, February 2004.
\newblock \href {https://doi.org/10.1023/B:APCS.0000013809.71153.30}
  {\path{doi:10.1023/B:APCS.0000013809.71153.30}}.

\bibitem{Batalin:1977pb}
Igor~Anatolievich Batalin and Grigory~Alexandrovich Vilkovisky.
\newblock Relativistic \({S}\)-matrix of dynamical systems with boson and
  fermion constraints.
\newblock {\em Physics Letters B}, 69(3):309--312, August 1977.
\newblock \href {https://doi.org/10.1016/0370-2693(77)90553-6}
  {\path{doi:10.1016/0370-2693(77)90553-6}}.

\bibitem{1981PhLB..102...27B}
Igor~Anatolievich Batalin and Grigory~Alexandrovich Vilkovisky.
\newblock Gauge algebra and quantization.
\newblock {\em Physics Letters B}, 102(1):27--31, June 1981.
\newblock \href {https://doi.org/10.1016/0370-2693(81)90205-7}
  {\path{doi:10.1016/0370-2693(81)90205-7}}.

\bibitem{1983PhRvD..28.2567B}
Igor~Anatolievich Batalin and Grigory~Alexandrovich Vilkovisky.
\newblock Quantization of gauge theories with linearly dependent generators.
\newblock {\em Physical Review D}, 28(10):2567--2582, November 1983.
\newblock \href {https://doi.org/10.1103/PhysRevD.28.2567}
  {\path{doi:10.1103/PhysRevD.28.2567}}.

\bibitem{BATALIN1984106}
Igor~Anatolievich Batalin and Grigory~Alexandrovich Vilkovisky.
\newblock Closure of the gauge algebra, generalized {L}ie equations and
  {F}eynman rules.
\newblock {\em Nuclear Physics B}, 234(1):106--124, March 1984.
\newblock \href {https://doi.org/10.1016/0550-3213(84)90227-X}
  {\path{doi:10.1016/0550-3213(84)90227-X}}.

\bibitem{Batalin:1985qj}
Igor~Anatolievich Batalin and Grigory~Alexandrovich Vilkovisky.
\newblock Existence theorem for gauge algebra.
\newblock {\em Journal of Mathematical Physics}, 26(1):172--184, January 1985.
\newblock \href {https://doi.org/10.1063/1.526780}
  {\path{doi:10.1063/1.526780}}.

\bibitem{cdb3b74a-b986-3a9f-9841-aad15f0544ec}
Jun{-i}chi Igusa.
\newblock A classification of spinors up to dimension twelve.
\newblock {\em American Journal of Mathematics}, 92(4):997--1028, October 1970.
\newblock \href {https://doi.org/10.2307/2373406} {\path{doi:10.2307/2373406}}.

\end{thebibliography}

\end{document}